\documentclass[sigconf,nonacm, authorversion]{acmart}

\usepackage[utf8]{inputenc}
\usepackage{graphicx}
\usepackage{xcolor}
\usepackage{color}
\usepackage{xspace}
\usepackage{hyperref}
\usepackage{booktabs}
\usepackage{tabularx}
\usepackage{tikz}
\usepackage{color}
\usepackage{soul}
\usepackage{amsmath}
\usepackage{mathtools}
\usepackage[linesnumbered,ruled, vlined]{algorithm2e}
\usepackage{cleveref}
\usepackage{xifthen}
\usepackage{pifont}
\usepackage{caption}
\usepackage{subcaption}
\usepackage{enumitem}
\usepackage{multirow}
\usepackage{makecell} 
\usepackage{xr}
\usepackage{wrapfig}

\makeatletter
\newcommand*{\addFileDependency}[1]{
  \typeout{(#1)}
  \@addtofilelist{#1}
  \IfFileExists{#1}{}{\typeout{No file #1.}}
}
\makeatother

\newcommand*{\myexternaldocument}[1]{%
    \externaldocument{#1}%
    \addFileDependency{#1.tex}%
    \addFileDependency{#1.aux}%
}
\myexternaldocument{supplementary} 

\copyrightyear{2023}
\acmYear{2023}
\setcopyright{none}

\acmISBN{} 
\acmDOI{} 
\startPage{1}

\definecolor{lightyellow}{RGB}{250, 250, 180}
\sethlcolor{lightyellow}

\DeclarePairedDelimiter{\ceil}{\lceil}{\rceil}

\newcommand{\R}[1]{\ensuremath{R(#1)}\xspace}

\newcommand{\SIedge}[1]{\ensuremath{s(#1)}\xspace}
\newcommand{\SInode}[1]{\ensuremath{S^{o}(#1)}\xspace}
\newcommand{\SIInnode}[1]{\ensuremath{S^{i}(#1)}\xspace}

\newcommand{\KInNode}[1]{\ensuremath{I(#1)}\xspace}
\newcommand{\KOutNode}[1]{\ensuremath{O(#1)}\xspace}
\newcommand{\Descendants}[2][]{%
  \ifthenelse{\equal{#1}{}}
    {\ensuremath{#2^{\downarrow}}\xspace}
    {\ensuremath{#2^{\downarrow}_#1}\xspace}
}
\newcommand{\FO}[1]{\ensuremath{FO(#1)}\xspace}
\newcommand{\LO}[1]{\ensuremath{LO(#1)}\xspace}

\newcommand{\ST}[1]{\ensuremath{ST(#1)}\xspace}

\newcommand{\W}[1]{\ensuremath{W(#1)}\xspace}

\newcommand{\SD}{\ensuremath{T_\infty^s}\xspace}

\newcommand{\tasks}{\ensuremath{V}\xspace}
\newcommand{\edges}{\ensuremath{E}\xspace}

\newcommand{\level}[1]{\text{L}\ensuremath{(#1)}}

\newcommand{\WCC}[1]{\ensuremath{WCC(#1)}}

\newcommand{\heurnewblock}{\textsf{SB-LTS}}
\newcommand{\heurnonewblock}{\textsf{SB-RLX}}
\newcommand{\ssched}{\textsf{STR-SCH}}
\newcommand{\schednewblock}{\textsf{\ssched-1}}
\newcommand{\schednonewblock}{\textsf{\ssched-2}}
\newcommand{\nssched}{\textsf{NSTR-SCH}}

\newcommand{\one}{\ding{192}\xspace}
\newcommand{\two}{\ding{193}\xspace}
\newcommand{\three}{\ding{194}\xspace}

\theoremstyle{definition}

\begin{document}
\title{{Streaming Task Graph Scheduling for Dataflow Architectures}}
\subtitle{This is a preprint version of the paper that will appear at ACM HPDC'23}

\author{Tiziano De Matteis, Lukas Gianinazzi, Johannes de Fine Licht, Torsten Hoefler}
\affiliation{
    \institution{ETH Zurich}
    \department{Department of Computer Science}
    \city{Zurich}
    \country{Switzerland}
}
\email{{tdematt, lukas.gianinazzi, definelicht, htor}@inf.ethz.ch}

\begin{abstract}
Dataflow devices represent an avenue towards saving the control and data movement overhead of Load-Store Architectures.
Various dataflow accelerators have been proposed, but 
how to efficiently schedule applications on such devices remains an open problem. The programmer can explicitly implement both temporal and spatial parallelism, and pipelining across multiple processing elements can be crucial to take advantage of the fast on-chip interconnect, enabling the concurrent execution of different program components.
This paper introduces \emph{canonical task graphs}, a model that enables \emph{streaming scheduling} of task graphs over dataflow architectures.
We show how a task graph can be statically analyzed to understand its steady-state behavior, and we use this information to partition it into temporally multiplexed components of spatially executed tasks.
Results on synthetic and realistic workloads show how streaming scheduling can increase speedup and device utilization over a traditional scheduling approach.
\end{abstract}

\maketitle

\section{Introduction}
The end of Dennard scaling and Moore’s law have breathed new life into the computer architecture research field, with researchers looking for alternatives to overcome the inherent inefficiencies of traditional Load-Store Architectures (LSAs).
Driven by the specific needs of application domains such as machine learning, various highly parallel computing platforms have recently been proposed to accelerate specific parts of, or even entire, computations.
These devices come in the flavor of Domain-Specific Architecture (such as Google TPUs for machine learning workloads, \cite{tpu}), devices with hardened logic but flexible datapaths (Configurable Corse Grain Array, CGRA, such as the Xilinx's ACAP platform, \cite{versal}), and large chips (such as the Cerebras Wafer Scale Engine, \cite{cerebras}).
%
They all are characterized by spatial parallelism, having tens to thousands of Processing Elements (PEs), and a fast Network-On-Chip (NoC) for efficient inter-PE communications.

Scheduling an application on these devices poses challenges different from traditional LSAs. First, in these architectures, the computation can be performed both \textit{spatially}, by taking advantage of a large number of computing units, and \textit{temporally}, by time-multiplexing resources to perform the computation. This introduces a trade-off that must be understood to optimally schedule an application on a given device.
Second, and strictly related, \textit{pipelining} can be crucial to fully exploit the device's spatial parallelism. It allows for the concurrent execution of different program components, exploiting fast on-chip communication while reducing off-chip memory accesses.

\begin{figure}[tbp]
\centering
\includegraphics[width=\columnwidth]{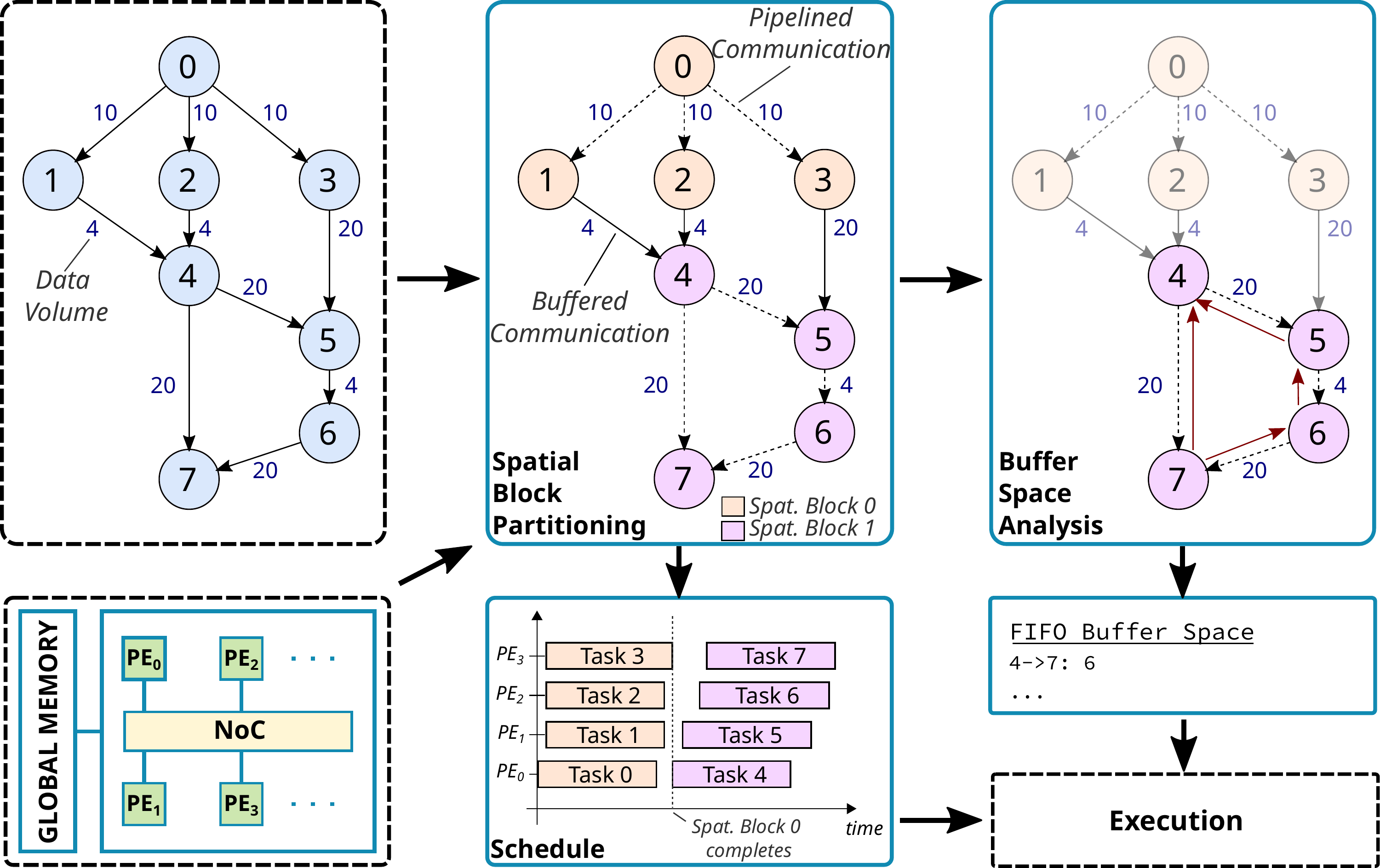}
\caption{Overview of the proposed approach. Blue solid boxes identify the various scheduling steps. Dashed boxes identify user-provided inputs and final output.}
\label{fig:dag}
\end{figure}

Synchronous DataFlow Graphs~\cite{synchronous_data_flow_graph} and extensions, such as Cyclo-static Dataflow Graphs \cite{csdfg}, have been proposed as models of computation for dataflow applications. However, these models focus on schedules that improve the throughput of multiple application iterations. 
In this work, instead, we tackle the issue of scheduling a direct acyclic task graph onto a dataflow device: we propose a methodology to analyze and optimize the application scheduling by explicitly considering the challenges described above, pipelining within a single application iteration and optimizing for latency. 
In particular, we:
\begin{itemize}
    \item introduce \emph{canonical task graphs}: a dataflow-centric view of the computation to model and statically analyze the execution of an application on an abstract dataflow architecture, taking pipelining into account;
    \item propose algorithms for scheduling the application, considering spatial and temporal multiplexing;
    \item derive bounds on the parallel execution time of task graphs;
    \item present algorithms to guarantee deadlock-free execution in the presence of pipelined tasks.
\end{itemize}

\section{Problem Definition}

The input application is described by a \textit{Direct Acyclic Task Graph} (DAG) $G=(\tasks,\edges)$ (see \Cref{fig:dag}, top-left box), where \tasks is the set of tasks in which the application can be decomposed, and an edge $(i,j) \in \edges$ represents a (data) dependency between task $i$ and task $j$. Edge labels represent the amount of data communicated between tasks, counted in unitary elements (e.g., floating point numbers). For simplicity, we assume that all edges transport basic data types, but the approach can be directly applied to any data width (i.e., edges can carry vectors of data).

We model the targeted dataflow device as consisting of $P$ homogeneous Processing Elements (PEs) and a global memory, which is assumed to have infinite size (\Cref{fig:dag}, bottom-left box). PEs can communicate directly with each other and with global memory.
They are interconnected by a Network-on-Chip (NoC) in which we assume all communications perform without contention. PEs can execute one task at a time.
Tasks are non-preemptive, and they can communicate according to two different modes:
\begin{itemize}
    \item \textit{Buffered} communication (solid edges in \Cref{fig:dag}). The producer stores the output data in global memory, and the consumer can later access the data.  The consumer can start only \textit{after} the producer has finished.
    \item \textit{Streamed} (or \textit{pipelined}) communication (dashed edges in \Cref{fig:dag}): the producer \textit{streams} the data to the consumer, element by element.
    The consumer can start \emph{as soon as} the producer outputs the first element, and their execution is (partially) overlapped.
\end{itemize}
Given an application, we want to derive a \textit{static} schedule for its execution on the target device. \Cref{fig:dag} outlines the proposed approach:
\begin{enumerate}
    \item The task graph is partitioned into temporally multiplexed components that we call \emph{spatial blocks}. Each of these components can have at most $P$ tasks.
    \item The spatial blocks are scheduled one after the other, respecting the dependencies expressed in the task graph. Tasks within a spatial block can executed simultaneously, exploiting spatial parallelism and, whenever possible, pipelining.
    \item The amount of buffer space needed to guarantee deadlock-free execution of pipelined tasks is computed.
\end{enumerate}
Buffered communication is the only communication mode allowed between tasks belonging to different spatial blocks. On the other hand, pipelined communication is the preferred way of communication within a streaming block, as it allow overlapping the execution of communicating tasks, and may reduce off-chip memory accesses.

For some dataflow architectures, such as CGRAs, locality and placement play an important role in achieving high performance. We do not explicitly deal with placement in this work, but we believe that the proposed approach can be the starting point for tackling a similar challenge.

Once all tasks in a graph are scheduled, the \emph{makespan} (i.e., the schedule length) is given by the maximum finishing time of any exit node of the graph.
\emph{Our goal is to find the graph partitioning and task-to-PE assignments that minimize the application makespan.}

\section{Canonical Task Graphs}\label{sect:canonical_dags}

In the following, we introduce \emph{canonical task graphs} that are DAGs composed of certain types of nodes and that respect specific rules.  Then we discuss how we can represent applications within this framework. Canonical task graphs are the subject of the analysis and proofs discussed in \Cref{sect:analysis}.

\subsection{Canonical nodes}

We define as \textit{canonical} a node that has a bounded number of input and output edges, that receives the same amount of data from all its input edges, and produces the same amount of data to all its output edges. 
A task graph that is composed of only canonical nodes is \textit{canonical}.
We distinguish between nodes that perform the actual computation and nodes that serve as buffers to store, replicate, and reshape data.

A computational node $v$ receives \KInNode{v} data elements from all the input edges and produces $\KOutNode{v} = \R{v}\KInNode{v}$ elements to each output edge. We call $\R{v}$ the node's \textit{production rate}. We distinguish between three notable cases:
\begin{enumerate}
    \item $\R{v}=1$, is the case of \textit{element-wise} nodes. Examples of element-wise tasks are vector-vector addition, Hadamard product, and various activation functions in machine learning models.
    \item $\R{v}<1 $, is the case of \textit{downsampler} nodes. Downsamplers can be used to represent reductions, such as dot product, statistics, pooling operators.
    \item $\R{v}>1$, for \textit{upsampler} nodes. Examples: vector concatenation, data replication. 
\end{enumerate}
In the following, we adopt a\textit{ dataflow-centric} view on the computation: we assume that the operations applied to compute over the input or output data elements require linear time (one time unit per element) and constant space. Furthermore, we assume PEs can always satisfy the rate requirements implied by the given task graph: they can accept/produce one element per unit of time from each task inputs/outputs. Other than this, we do not restrict the internal node semantics. \Cref{sect:generic_dags} discusses how more complex operations, such as outer products or vector normalization, can be represented as canonical subgraphs, capturing their actual compute time and dataflow.


We define \textit{buffer} nodes as follows: a buffer node $v$ with a production rate $\R{v}$ buffers its inputs, and once \emph{all} input elements have been stored, they are output $\R{v}$ times. It follows that we cannot pipeline communications \textit{through} a buffer node. Buffer nodes are not active entities: unlike computational nodes, they do not need to be scheduled on a PE. They can be implemented as memory components, such as backing global memory or cache/scratchpad area in the PEs. We introduce them to model and analyze computations that may require reading the data multiple times or in a different order. For example, the buffer node can output $\R{v}$ copies of the input, or a reshaped version of the input.

A source node $v$ reads its output from global memory. It does not have a production rate, and directly outputs $\KOutNode{v}$ elements. A sink node $v$ stores its inputs in global memory, and its production rate is zero.

\subsection{Representing generic computations with canonical task graphs}\label{sect:generic_dags}

Canonical nodes allow the user to represent a meaningful set of operations. However, in generic computations, tasks may receive/produce arbitrary data volumes from/to their incident edges. This is the case for operations such as outer product and matrix multiplication.
In the following, we go through some of these examples, discussing how we can map them into canonical task graphs, correctly modeling their computation time and parallelism opportunities, if any.

How an operation can be represented in a canonical task graph depends on its actual implementation and runtime behavior. While this can impact the parallelism and streaming opportunities, a more general case analysis is outside the scope of this work, as this would require complete knowledge of the internal semantics of the tasks to capture their data access patterns. For the rest of the paper, we assume the canonical task graph is provided as a result of a compiler or synthesis pass on the user’s program.

\subsubsection{Outer Product}\label{sect:outer-product}
Let us consider the outer product between an $N$-element vector $u$ and an $M$-element vector $v^T$, that produces a matrix $A$ of size $N\times M$. 
\figurename~\ref{fig:outer_product} shows various implementations and representations for the outer product operation as a canonical task graph. All edge widths are equal to one (i.e., edges carry scalar values).
\begin{figure}[tb]
    \centering
    \includegraphics[width=\columnwidth]{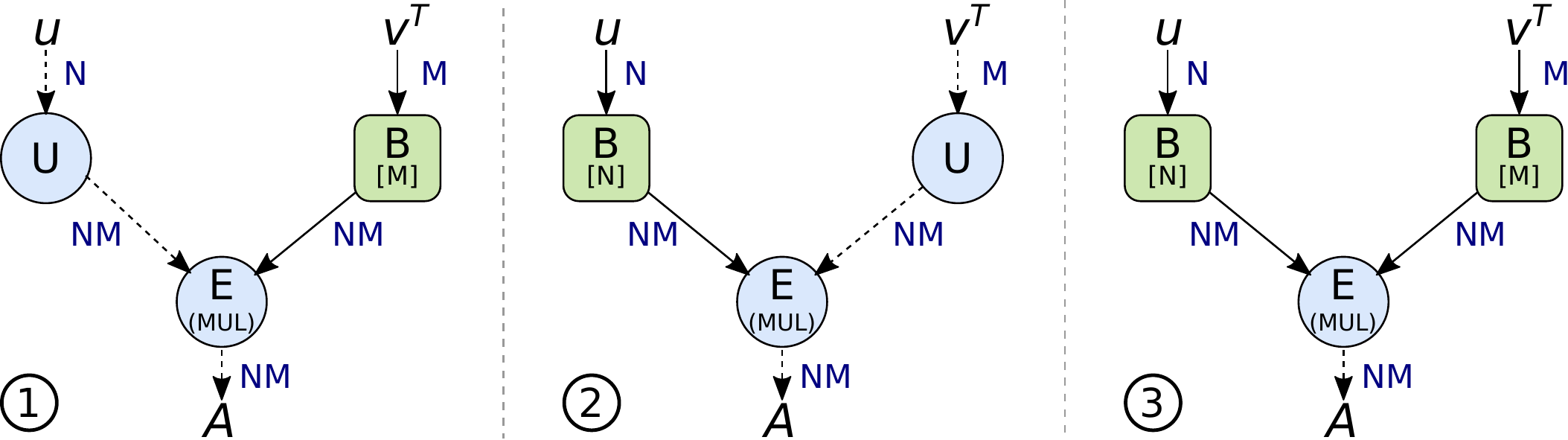}
    \caption{Outer product implementations and their representations as canonical graphs. Green squared nodes represent buffer nodes, with the buffer space in squared brackets.}
    \label{fig:outer_product}
\end{figure}
Task graph \one considers the case where every element of $u$ is multiplied by the entire vector $v^T$, producing the final matrix $A$ row-by-row. 
In this case, every element of $u$ is replicated $M$ times through an upsampler. Vector $v^T$ must instead be read entirely $N$ times, and therefore is stored in a buffer node.
An element-wise node represents the actual multiplication of $u$ and $v^T$ elements, taking $NM$ elements from both its inputs and producing $NM$ elements in output.
This implementation allows the streaming of input vector $u$ and outputs matrix $A$ in row-major order.
Implementation \two shows the symmetric implementation where every element of $v^T$ is multiplied by the entire vector $u$. This will produce the output matrix $A$ column-by-column, enabling the streaming of $v^T$ and $A$.
Finally, implementation \three shows the case where both inputs are buffered. In this case, only the result can be streamed.

\subsubsection{Matrix-Matrix Multiplication}
Consider the case of a Matrix-Matrix multiplication $C=AB$, where $A$ is a matrix of size $N\times K$, $B$ is an $K\times M$ matrix, and $C$ is an $N\times M$ matrix. 
\begin{figure}[t]
\includegraphics[width=\columnwidth]{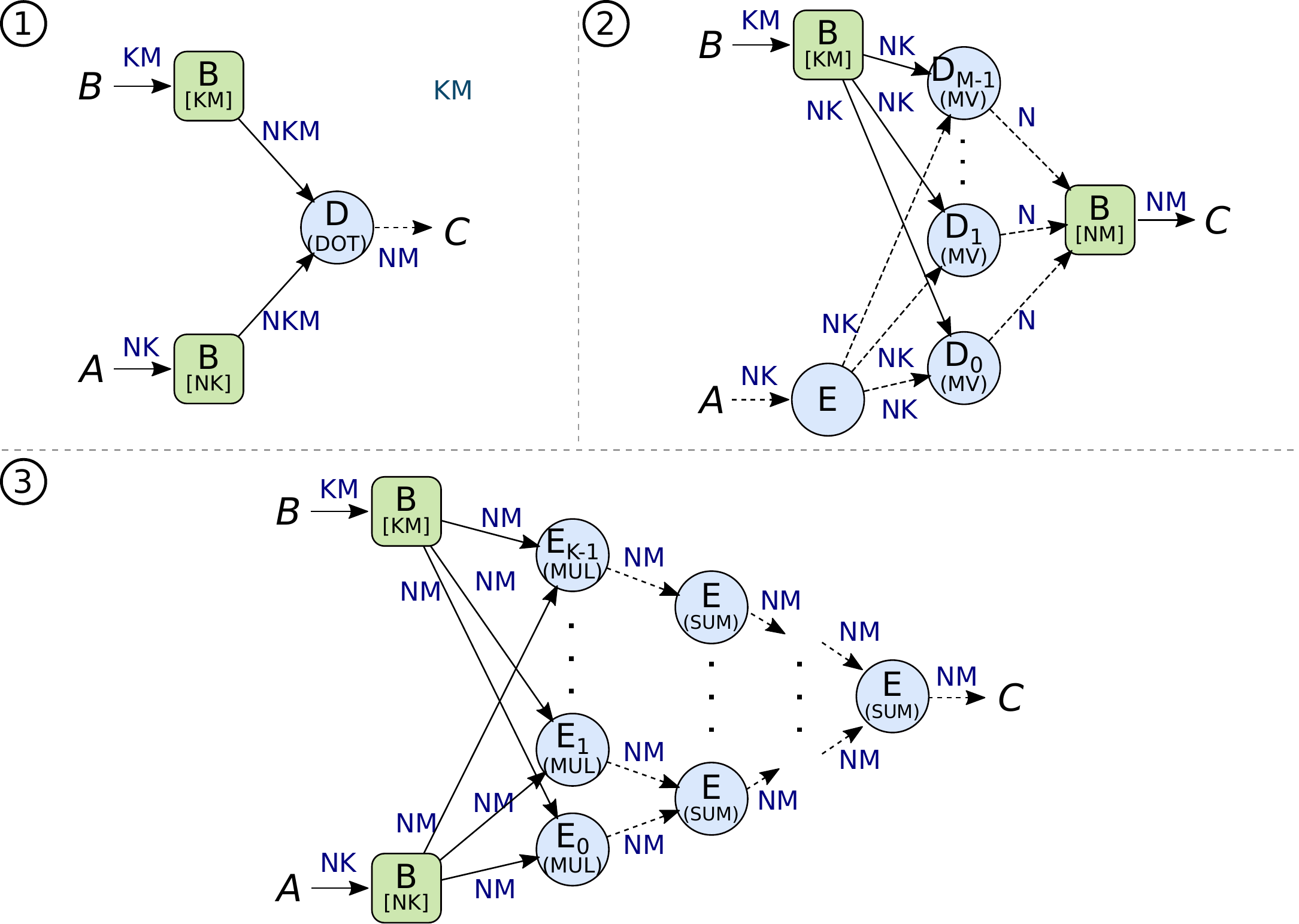}
\caption{Different Matrix-Matrix multiplication implementations and corresponding canonical task graphs.}
\label{fig:mmm_dag}
\vspace{-1em}
\end{figure}
In the naive inner product implementation, a row of $A$ is multiplied by a column of $B$ to produce a single element of $C$. Its canonical representation is shown in the graph \one of \Cref{fig:mmm_dag}.
In this case, we exploit two buffer nodes to replicate the two input matrices. The multiplication operation is represented with a downsampler node, with a production rate of $1/K$.

Usually, programmers are interested in parallel implementations of matrix-matrix multiplication.
Assuming that matrix $A$ elements arrive row-by-row, the task graph \two in \Cref{fig:mmm_dag} shows an implementation that computes each column of $C$ in parallel. Each task $D_i$ implements a matrix-vector multiplication: it is a downsampler task that takes in input the matrix $A$ and a column of $B$ to produce a column of $C$ ($N$ output elements). In this implementation, the input matrix $A$ can be streamed directly to the various computing tasks. The left-topmost task behaves like an element-wise operation by replicating its input elements to the output edges. On the contrary, elements of matrix $B$ must be read multiple times and need to be explicitly buffered. As it will become clearer in \Cref{sect:analysis}, under certain circumstances (if $K>M$), we can also stream the output row-by-row without performance penalties instead of buffering it.

Finally, task graph \three shows an implementation that parallelizes along the $K$ dimension: each inner task $E_i$ computes an outer product between a column of $A$ and a row of $B$. The outer product result is output row-by-row. The tree-structured rightmost element-wise tasks are in charge of performing the element-wise sum of the outer product results. The output matrix $C$ can be streamed to successive tasks. Similarly to \Cref{sect:outer-product}, we can derive an alternative implementation that produces the result column-by-column.

\subsubsection{Vector normalization}Let $x$ be an $N$ elements vector, we want to represent the vector normalization $y=\frac{x}{||x||}$ as a canonical task graph.
\begin{figure}[htp]
    \centering
    \includegraphics[width=0.8\columnwidth]{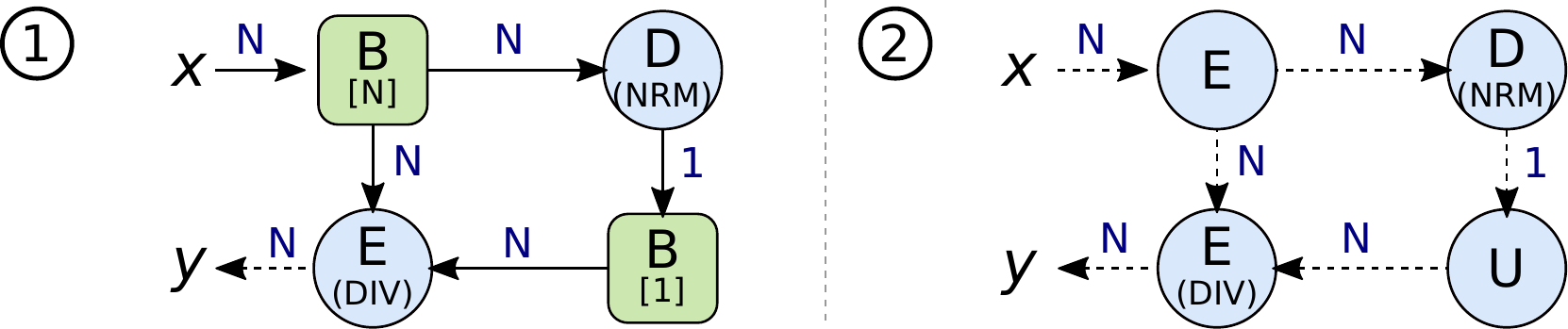}
    \caption{Vector normalization: different implementations and representations as canonical task graphs.}
    \label{fig:vector_norm}
\end{figure}
\Cref{fig:vector_norm} shows two different implementations. In \one, the vector $x$ is stored in a buffer since it must be read two times: one to compute its norm  (by the downsampler node) and one to compute the division between each element and the norm (by the element-wise node). Once computed, the norm value is stored in a buffer and read $N$ times to perform the divisions. With this implementation, no streaming communication can be exploited, and the two operations are executed one after the other.
When dealing with a sequence of vectors, we can double the buffer space ($2N$) so that while the element-wise node works on vector $i$, the downsampler can compute the norm for vector $i+1$.
Task graph \two shows the case where the input vector $x$ is not buffered, but instead streamed directly to both the downsampler and the element-wise nodes. As it will become more evident in \Cref{sect:buffer_space}, such a solution would require properly dimensioned buffer space for pipelined communications to prevent deadlocks.

\subsubsection{Softmax}\label{suppl:softmax}

The numerically stable softmax operation, applied over an input vector $x$ with $N$ elements, is defined as:
$$y_i = \frac{e^{{x_i}-\max{(\mathbf{x})}}}{\sum_{j=1}^N e^{{x_j}-\max{(\mathbf{x})}}}$$

\figurename~\ref{fig:stable_softmax} shows a possible implementation as a canonical task graph, where a separate task expresses each sub-operation. The leftmost downsampler computes the maximum value of vector $x$. This partial result is buffered and then used by the following three computational tasks to compute the value at the denominator: these are in charge respectively of subtracting to each value of $x$ the max (element-wise task), of exponentiating it (element-wise task), and sum the corresponding values (downsampler task). This value is then buffered and used by the bottom element-wise division task to compute the final outputs $y_i$.
\begin{figure}[htp]
    \centering
    \includegraphics[width=0.8\columnwidth]{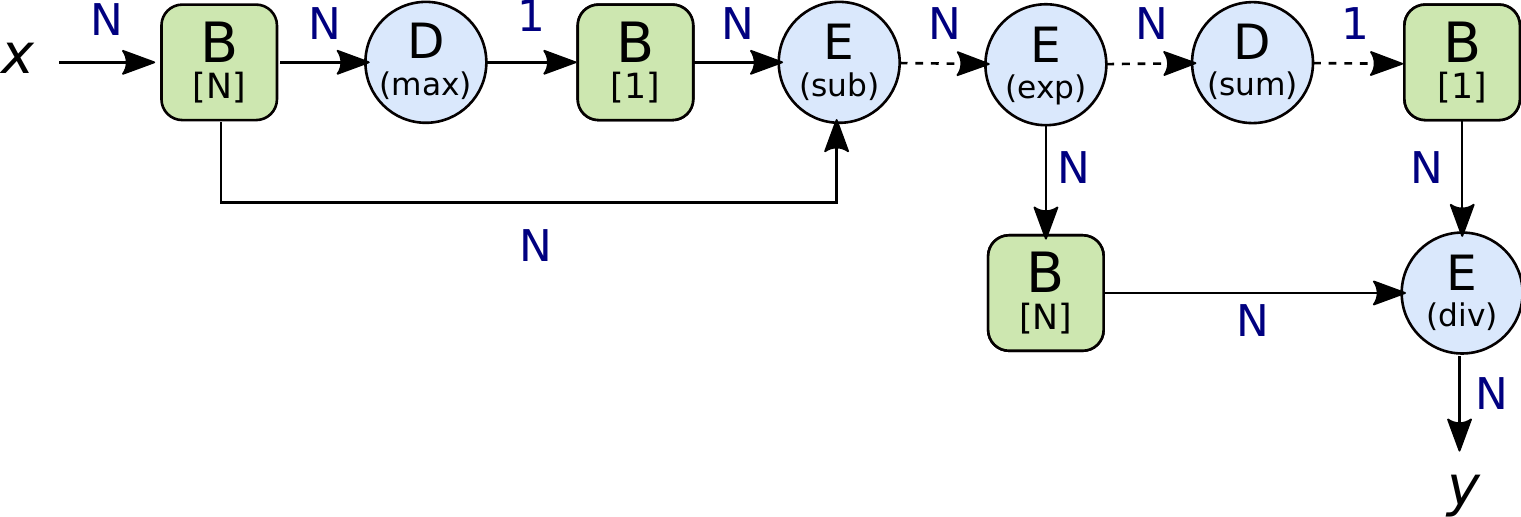}
    \caption{Softmax: representation as canonical task graph.}
    \label{fig:stable_softmax}
\end{figure}
Note that in this implementation the values $e^{x_i-\max({\mathbf{x})}}$ are computed only once and used both for computing the denominator and for the final division. This allows to partially stream the internal computation.

\section{Analysis}\label{sect:analysis}

Given a canonical task graph, we first characterize its execution for an infinite number of processing elements: each task can be assigned to a different PE, and the application is fully executed in the spatial domain. We analyze the application's \emph{steady-state} behavior and define its \textit{streaming depth}, analogous to the concept of depth of a parallel program. These analysis tools allow us to define the schedule of a given canonical DAG considering a limited number of PEs (Section \ref{sect:schedule}). In the following, we use the term task and node interchangeably. For reference, \Cref{tab:symbols} lists all the used symbols throughout the rest of the paper.

\begin{table}[htb]
	\small
	\begin{tabular}{ll}
		\toprule
		\textbf{Symbol} & \textbf{Meaning}                     \\ 
		\toprule
		$\R{v}$ & Production rate of the node $v$      \\
		$\SIedge{e}$ & Streaming interval of the edge $e$      \\
		$\SInode{v}$/ $\SIInnode{v}$ & Output / Input streaming interval of the node $v$      \\
		$\KInNode{v}/\KOutNode{v}$ & Number of elements read/produced by node $v$ \\
		\LO{v}      & The time the last element leaves node $v$. \\
		\FO{v}      & The time the first element leaves node $v$.\\
		\ST{v}      & The starting time of node $v$.\\ 
		\level{u}   & Level of the node $u$. \\
		\W{v}       & Work of the node $v$. \\
		\SD         & Streaming depth of the graph.\\
		\bottomrule
		
	\end{tabular}
	\caption{List of used symbols.}
	\label{tab:symbols}
\end{table}

\subsection{Streaming intervals}\label{sect:streaming_intervals}
To analyze a streaming (sub-)graph, we are interested in its steady-state behavior that is reached when for each edge $(u,v)$, the average interval between elements output by node $u$ and the interval between elements consumed by the following node $v$ are balanced. 
In the following, we discuss the conditions that apply when such a state is reached.

We define for each edge $e$ its \emph{streaming interval} \SIedge{e} as the average time interval between elements going through the edge $e$ while the edge is streaming. 
\begin{figure}[tb]
    \centering
    \includegraphics[width=\columnwidth]{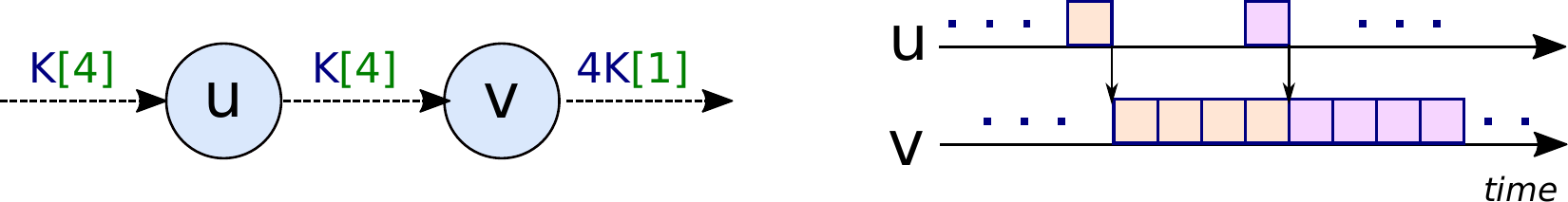}
    \caption{Streaming intervals. Left: the task graph with annotated data volumes (blue labels), and streaming intervals (green labels). Right: tasks data production over time.}
    \label{fig:streaming_interval_idea}
    \vspace{-1em}
\end{figure}
This concept is exemplified in \Cref{fig:streaming_interval_idea}. Here, task $v$ is an upsampler node with a production rate of $4$, which can produce a new output element on every time unit. At steady-state ($K\rightarrow \infty$), task $u$ can only produce an output element every four time units, blocked by the upsampler ingestion rate.

At the steady-state, a node will read data at an interval given by the maximum streaming interval of its incident input edges. It follows that given a node $v$, all its incident input edges will have the same streaming intervals, and all its output edges will have the same streaming interval. Hence, we denote the streaming interval of the input and output edges of $v$ with \SIInnode{v} and \SInode{v}, respectively. We continue with characterizing the streaming intervals.

First, all streaming intervals \SIedge{e} must satisfy 
\begin{equation}
\SIedge{e}\geq 1, 
\label{eqn:str_interval_requirement}
\end{equation}
where 1 represents a single unit of time (e.g., a clock cycle).
Second, each type of computational node introduces additional constraints, as follows.
Consider a computational node $v$, then its streaming intervals satisfy
\begin{equation}
\SInode{v}=\frac{\SIInnode{v}}{\R{v}} 
\label{eqn:streaming-interval}
\end{equation}
This implies that an element-wise node produces data at the same interval it receives. The output streaming interval of a downsampler is higher than the input one, while an upsampler node behaves the opposite.
Unlike computational nodes, buffer, source, and sink nodes do not affect the streaming intervals. They produce or ingest data at the smallest streaming interval the descendants can read, or the predecessors can produce. In addition, buffer nodes model communications that cannot be pipelined: they start producing data only once all the input elements have been received.

To compute the streaming intervals of a graph, we consider a transformed task graph where each buffer node (if any) is duplicated so that it occurs twice: as the sink of its predecessor nodes (\textit{tail}), and as the source of its successor nodes (\textit{head}). This allows us to capture that we cannot
stream through buffer nodes.
Then, we partition this graph into weakly connected components. We define as \WCC{v} the set of nodes in the same weakly connected component as $v$. All the communicating nodes in the same $WCC$ can exploit pipelined communications. 
Then, we can characterize the streaming interval of a node as follows:

\begin{theorem}\label{thm:si-single-source}
The streaming interval \SInode{v} of a node $v$ is $$\SInode{v} =\frac{\displaystyle \max_{u \in \WCC{v}} \  \KOutNode{u}}{\KOutNode{v}}$$
\end{theorem}
An example of a canonical task graph with annotated streaming intervals is shown in~\Cref{fig:streaming_interval_dag}. 
\begin{figure}[tb]
    \centering
    \includegraphics[width=\columnwidth]{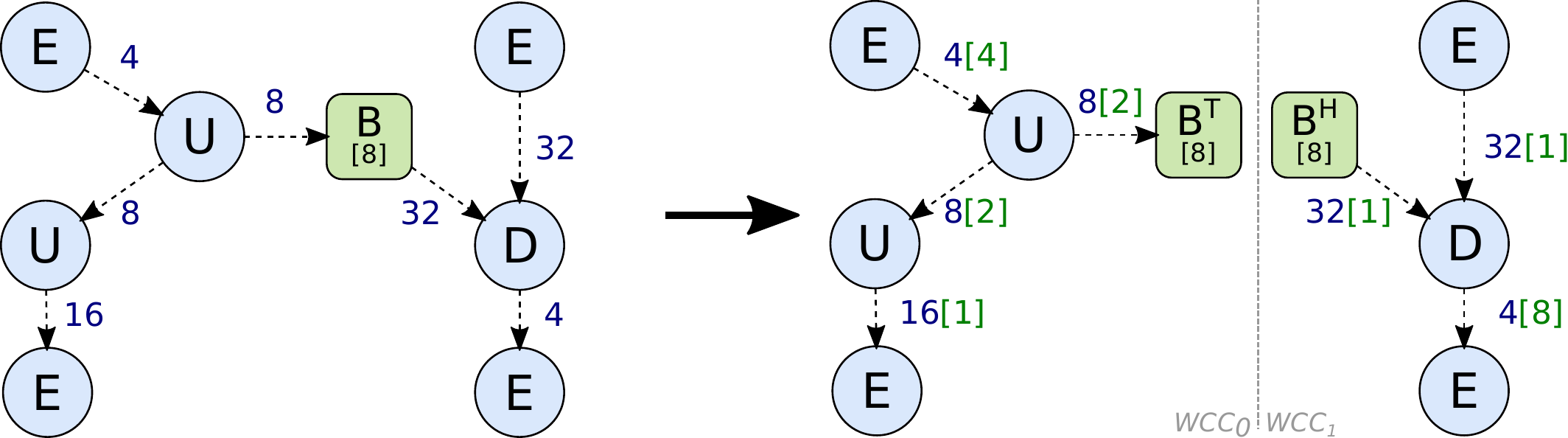}
    \caption{Streaming interval computation on a given task graph. Blue edge labels indicate data volume, green square bracketed labels indicate the computed streaming interval.}
    \label{fig:streaming_interval_dag}
\end{figure}
The graph is transformed into a new one composed of two $WCC$s, whose edge streaming intervals are independent of each other.

To prove \Cref{thm:si-single-source}, we introduce two lemmas that relate the streaming intervals and the input/output elements.

\begin{lemma}\label{lem:streaming-inverals-single-source}
If there is a path from a node $v_i$ to a node $v_j$ that does not contain a buffer node, the input streaming interval of $v_j$ satisfies $\SIInnode{v_j} = \frac{ \KOutNode{v_i}}{\KInNode{v_j}} \SInode{v_i}$.
\end{lemma}
\begin{proof}
The proof is by induction on the number of edges in the path. If there is a single edge in the path, $\KOutNode{v_i}=\KInNode{v_j}$, and the statement is trivial. Otherwise, consider the predecessor $u$ of $v_j$ in the path. By induction hypothesis, $\SIInnode{u} =  \frac{ \KOutNode{v_i}}{\KInNode{u}} \SInode{v_i}$. By \Cref{eqn:streaming-interval} and $\KInNode{u}=\KInNode{v_j}/\R{u}$, the result follows.
\end{proof}

Next, we argue about what happens in a WCC that does not contain buffer nodes, but possibly has multiple sources. 
\begin{lemma}\label{lem:streaming-intervals-multiple-source}
For all pairs of nodes $v_i \neq v_j$ where there exists a node that is reachable from both $v_i$ and $v_j$ without using buffer nodes along the way, we have that $\SInode{v_i} \KOutNode{v_i} = \SInode{v_j} \KOutNode{v_j}$.
\end{lemma}
\begin{proof}
Consider two such pairs $i\neq j$. Consider two paths from $v_i$ and $v_j$ that end in a node $v'$.
By applying \Cref{lem:streaming-inverals-single-source} to both paths, we see that the streaming intervals satisfy $\frac{\KOutNode{v_j}}{\KInNode{v'}}\SInode{v_j}= \SIInnode{v'}=\frac{\KOutNode{v_i}}{\KInNode{v'}}\SInode{v_i}$. 
Hence, we conclude that $\KOutNode{v_i} \SInode{v_i}  = \KOutNode{v_j}\SInode{v_j}$. 
\end{proof}

This means that the product of the number of output elements and the output streaming interval is constant for all pairs of nodes in a WCC, which implies \Cref{thm:si-single-source}:
\begin{proof}[Proof of \Cref{thm:si-single-source}]
Consider the vertex $u$ in the weakly connected component with the largest number of output elements $\KOutNode{u}$. 
By \Cref{lem:streaming-intervals-multiple-source}, the streaming interval of a vertex $v \in \WCC{u}$ is $\SInode{v}=\frac{\SInode{u} \KOutNode{u} }{\KOutNode{v}}$. For the theorem to follow, the output streaming interval of $u$ has to be $\SInode{u}=1$. This ensures that all other streaming intervals are at least $1$. The streaming interval $\SInode{u}$ cannot be smaller than $1$ because that would violate \Cref{eqn:str_interval_requirement}. If it was larger than $1$, we could divide all streaming intervals in the WCC by this value and obtain a faster execution. Hence, we conclude $\SInode{v}=\frac{\KOutNode{u} }{\KOutNode{v}}$, as claimed.
\end{proof}

\Cref{thm:si-single-source} give us a linear-time, in the number of nodes, algorithm to compute the streaming intervals by just considering the maximum volume of data produced by the nodes in the weakly connected components.


\subsection{Work and depth analysis}
The \textit{work and depth} model is used to analyze the running time of parallel algorithms independently of the execution platform~\cite{bib:blelloch_work_depth}.
The algorithm's cost is determined by considering the \textit{work}, i.e., the total number of operations that are performed, and the \textit{depth}, i.e., the length of the longest shortest sequence of operations from any input to any output.

Under the assumption that the operations applied to compute over the input or output data elements require linear time (one time unit per element) and constant space, we define the work of a node as follows:

\begin{description}[style=unboxed,leftmargin=0em]
    \item[Work of a node] Given a node $v$ with \KInNode{v} input and \KOutNode{v} output items, we define its work as $\W{v}=\max\{\KInNode{v}, \KOutNode{v}\}$.
\end{description}
The work relates to the ideal execution time of the node in isolation. That is, how much time the node takes to consume and produce all the input and output data.

\begin{description}[style=unboxed,leftmargin=0em]
\item[Work of the graph]The work of a graph $G$ is defined as the sum of the work of its nodes, and it is equal to the execution time of the DAG on a single processor,  $T_1=\sum_{v\in V(G)} \W{v}$.
\end{description}
If the task graph comprises only basic operations, this is equivalent to the definition of work according the circuit work and depth model~\cite{bib:blelloch_work_depth}.

To characterize the streaming execution of a canonical task graph,  we introduce the concept of \textit{Streaming Depth}.

\begin{description}[style=unboxed,leftmargin=0em]
\item[Streaming Depth] We define as \textit{Streaming Depth} (\SD) the minimum time needed to perform the computation with an infinite number of PEs, when all computational tasks can be co-scheduled, and they can stream.
\end{description}

In the following, we start by looking at graphs that have only certain type of nodes, and then we
consider the general case, for which the streaming intervals play a crucial role.

\subsubsection{Element-wise graph}\label{sect:elem-wise-strdepth}
Let us assume that the task graph is composed by a single connected component and $N$ element-wise tasks, reading and producing $k$ elements. The streaming interval for each edge is $1$ (\Cref{thm:si-single-source}).
The work $T_1$ is given by $T_1=Nk$. 
Formally, we define the \emph{level} of a node $v$ in the task graph $G$ as:
\begin{align*}
\level{v} =
\begin{cases}
1 \hfill \text{if v has no parent,} \\
\max_{(u,v)\in E(G)}  \level{u} + 1  \ \ \ \text{else.}
\end{cases}
\end{align*}
The \emph{number of levels} in the task graph $G$ is defined as the maximum level of any of its vertices, $\level{G}=\max_{v \in V(G)} \level{v}$.

Given an infinite number of PEs, the time to execute $G$
depends on the time it takes to inject all the data ($k$) and the time it takes for the last element to leave a source/input node to reach a sink/output node (the number of levels in $G$ minus 1). 
Therefore the \emph{streaming depth} of $G$ is $\SD = k+\level{G}-1$.
Note that if the task graph $G$ is executed \emph{without streaming}, its depth would be $k \cdot \level{G}$. Hence, the deeper the task graph, the bigger the advantage that streaming might provide.

\subsubsection{Downsamplers graph}
In a DAG composed only of element-wise and downsampler nodes, the source(s) produce(s) the maximum number of elements. Once the source(s) generate(s) all the data, this must traverse all the graph to reach the sink(s) of the graph. Therefore, we can define the streaming depth of the graph as a generalization of the element-wise DAG:  $\SD = (\max_{v \in V(G)} \W{v} ) + \level{G} - 1$.

\subsubsection{General Canonical DAG}\label{sect:general_dag_analysis}

We first consider the case of a single weakly connected component $G$ without buffer nodes. 
We generalize the notion of levels: the level $\level{v}$ of a node is given by:
\begin{align*}
\level{v} =
\begin{cases}
1 \hfill \text{if v has no parent,} \\
 max (\R{v}, 1) + \max_{(u, v)\in E(G)} \level{u} \ \ \ \text{else.}
\end{cases}
\end{align*}

This is the time it takes for the last element leaving a source node to reach node $v$ and be processed, taking into account the presence of upsampler nodes.
As before, the \emph{number of levels} $\level{G}$ of a graph $G$ is the maximum $\level{v}$ over any of its vertices $v$.

Let the last-out time \LO{v} be the time the last element leaves node $v$. 
If $v$ is neither a source nor a buffer node, then
\begin{small}
\begin{align}\label{eq:lo}
\LO{v} = \max_{(u, v)\in E(G)} \LO{u}&+
\begin{cases}
\lceil(\R{v}-1) \SInode{v}\rceil +1 & \text{if $\R{v}>1$} \\
1  &\text{else.} 
\end{cases}
\end{align}
\end{small}
That is, element-wise and downsampler nodes finish their execution once they receive the last element from all the predecessors and produce the corresponding result. Upsampler nodes need more time as they have have to produce more than a single output element.
For a source node $v$ we have:
$$ \LO{v} = \lceil (\KOutNode{v} -1) \SInode{v} \rceil + 1  \enspace,$$ 
which implies that
\begin{align*}
(\KOutNode{v}-1) \SInode{v} < \LO{v} \leq \KOutNode{v} \SInode{v}  \enspace,
\end{align*}
as the streaming interval may not be an integer number in the general case. 
From \Cref{thm:si-single-source} we know that a source node's streaming interval depends on the maximum amount of data produced in the same weakly connected component.
Therefore, we can write the streaming depth as: 
\begin{align}\label{eqn:streaming-depth-wcc}
    \SD\leq  \level{G} + \max_{u\in G} \KOutNode{u}   \enspace ,
\end{align}
where this bound is exact as the number of elements being streamed goes to infinity.
Note how, it suffices to look at the data volumes and the number of levels, as the actual streaming intervals can be written in terms of the data volume.

Let us now consider task graphs comprising also buffer nodes.
If $v$ is a buffer node we first need to wait for the completion of all previous task before starting to produce new data. Therefore, its last-output time is defined as:
 $$\LO{v} = \max_{(u, v)\in E(G)} \LO{u} + \lceil (\KOutNode{v} -1) \SInode{v} \rceil + 1 \enspace.$$

To bound the streaming depth in general, we require the following constraint on how buffer nodes are placed in a canonical graph: after ignoring the directions of edges between pairs of non-buffer nodes, no directed cycle contains a buffer node. Such cycles would create the need for large "implicit"  buffers and can always be avoided by introducing an additional buffer that breaks the cycle.
Under this assumption, we decompose the graph into weakly connected components by splitting the buffer nodes, as described in \Cref{sect:streaming_intervals}. Then, we can apply the bound from \Cref{eqn:streaming-depth-wcc} to each of those components to determine their streaming depth. We create a new DAG $H$ by merging each WCC into a \emph{supernode}, and create an edge between each pair $(u, v)$ of supernodes where $u$ contains the tail and $v$ contains the head of a split buffer node. Each supernode is assigned a depth equal to the depth of the WCC it represents. Then, we compute the depth $T_{\infty}(H)$ of $H$ as the deepest path in $H$. Note that $H$ is acyclic because of our requirement on how buffers are placed in canonical DAGs.
In general, it could be beneficial to start running another WCC even though not all nodes in its preceding WCC have finished. 
However, as the number of elements being streamed goes to infinity, this bound becomes tight.
Let $\hat L$ be the largest total number of levels in the WCCs given by any source to sink path in $H$. Then, the bound satisfies $\SD(G) \leq T_{\infty}(H)\leq \SD(G) + \hat L$. 

\section{Scheduling}\label{sect:schedule}
If there are more tasks than PEs ($P < N$), the graph must be partitioned in temporally multiplexed components of at most $P$ spatially executed tasks. In the following, we refer to such components as \emph{spatial blocks}.  As each spatial block is co-scheduled, all edges between computational tasks of the same spatial block can be streaming edges. On the contrary, edges between spatial blocks are non-streaming. The partitioning must be done so that the overall execution time ($\max_{v\in V} LO(v)$) is minimized.
\Cref{sect:general_dag_analysis} shows how the last-out time relates to the maximum amount of data produced by the nodes in the graph.
This allows us to define the scheduling problem as an optimization problem.
Given a canonical task graph, we want to partition it into spatial blocks containing at most $P$ computational nodes, such that:
\begin{itemize}
\item the sum of the maximum data volume being read or produced by a node of each spatial block is minimized;
\item the dependencies between spatial blocks still form an acyclic graph and respect the original task graph semantic (the graph induced by a spatial block is still acyclic, being a subgraph of the original task graph).
\end{itemize}

In the following, we first discuss how to schedule the tasks within a given spatial block. Then we discuss the case of general canonical task graphs, outlining an heuristic for the spatial block partitioning.

\subsection{Scheduling within a spatial block}\label{sect:schedule_in_sb}
All the tasks in a spatial block can be co-scheduled (similarly to gang-scheduling~\cite{gang_scheduling}), and take advantage of pipelined communications. 
Let be $B_i$ the current spatial block and $G[B_i]$ the subgraph induced by the tasks of $B_i$. When we schedule tasks in the spatial block $B_i$, all tasks in the spatial block $B_{i-1}$ have completed. 

We define as \textit{first-out time }\FO{v} the time the first element leaves node $v$. 
If $v$ is not a buffer node or a source of the block:
\begin{footnotesize}
\begin{align*}
\FO{v} &= \max\limits_{(u, v)\in E(G[B_i])} \FO{u} + 
\begin{cases}
\left\lceil\left(\frac{1}{\R{v}}-1\right) \SIInnode{v}\right\rceil +1 & \text{if $\R{v}<1$} \\
1  &\text{else.} \\
\end{cases}
\end{align*}%
\end{footnotesize}
Element-wise and upsampler nodes can output the first element as soon as they receive one from all predecessor nodes. On the other hand, downsampler nodes must accumulate data according to their given production rate before producing the first result.
If $v$ is a buffer node, then it needs to wait for the completion of all preceding tasks and its first-out time can be defined as:
\begin{align*}
\FO{v} =
\max_{(u, v)\in E(G[B_i])} \LO{u} + 1 
\end{align*}

If $v$ is a source node of the block, but not of the graph, it has to wait for the completion of tasks in previous blocks:
\begin{small}
\begin{align*}
\FO{v} = \max_{(u, v)\in E(G)} \LO{u} &+
\begin{cases}
\left\lceil\left(\frac{1}{\R{v}}-1\right)\frac{\SInode{v}}{\R{v}}\right\rceil +1 & \text{if $\R{v}<1$} \\
1  &\text{else.} \\
\end{cases}
\end{align*}
\end{small}
Finally, $\FO{v}=1$ if $v$ is a source of the whole task graph.

When a node produces some data, its streaming successor (if any) is ready to start.
We define the starting time of a task $v$ as follow:
\begin{small}
\begin{align*}
\ST{v} = 
\begin{cases}
0  & \text{if $v$ is a source of the graph} \\
\max_{(u,v)\in G}\LO{u}  &\text{if $v$ is a source of the block} \\
\max_{(u,v)\in G[B_i]}\FO{u} & \text{otherwise.}
\end{cases}
\end{align*}
\end{small}
Once scheduled, a task will run until its last-out time (see \Cref{eq:lo}). \Cref{fig:schedule_example} shows an example of spatial block scheduling. It is worth remarking that buffer nodes are not active entities: they are used to express buffering opportunities and will be not actually scheduled on a PE. Anyway, they play a crucial role in scheduling as they affect the first-output and last-output times of successive tasks.

\begin{figure}[t]
\begin{minipage}{.45\columnwidth}
\centering
\includegraphics[width=0.7\textwidth]{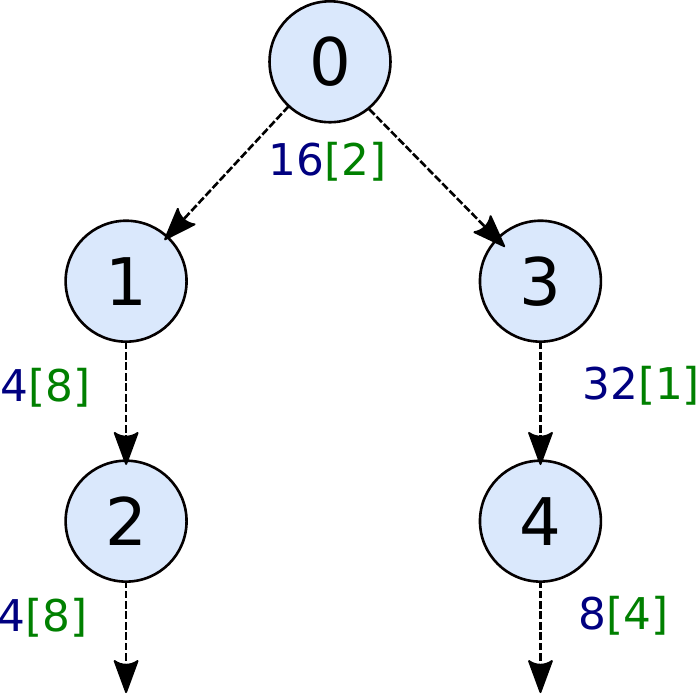}
\end{minipage}
\quad
\begin{minipage}{.45\columnwidth}
\footnotesize
\begin{tabular}{cccc}
\toprule
\textbf{Task} & $\mathbf{ST}$ & $\mathbf{LO}$ & $\mathbf{FO}$ \\
\midrule
0    & 0     & 31  & 1  \\
1    & 1     & 32  & 8  \\
2    & 8     & 33  & 9   \\
3    & 1     & 33  & 2   \\
4    & 2     & 34  & 6   \\
\bottomrule
\end{tabular}
\end{minipage}
\caption{A spatial block (left) and its schedule (right).}
\label{fig:schedule_example}
\vspace{-1em}
\end{figure}

To schedule a task, we should look at all its predecessors to compute its starting, first-out, and last-out time. Computing the streaming intervals requires linear time in the number of nodes (see \Cref{sect:streaming_intervals}). Therefore, scheduling all the tasks in the spatial blocks requires $\mathcal{O}(N^2)$, where $N$ is the number of nodes in the graph.

\subsection{Spatial block partitioning}\label{sect:scheduling_generic_task_graph}
The NP-hard \textit{sum-of-max partition problem under a Knapsack Constraint}~\cite{sum-of-max} is equivalent to the spatial block partitioning  problem of a generic canonical task graph where the spatial blocks are restricted to be connected components. 

We propose a greedy heuristic for the spatial block partitioning, outlined in \Cref{alg:streaming_blocks_generic}.
\begin{algorithm}[htb]
	\footnotesize
	\KwIn{General canonical task graph $G = (V,E)$, number of PEs $P$, variant $\in$ \{\heurnewblock, \heurnonewblock\}}
	\KwResult{a partition of $G$'s nodes in spatial blocks}
	$SB \gets [\{\emptyset\}]$; $i \gets 0$\;
	\While{$|V|>0$}{
	    $sources \gets$ compute source nodes of the graph\;
	  	    
	    $cand \gets$ node in $sources$ producing less data than than the block's sources, a node which is a block source, or the node in $source$ producing less data (if variant $=$ \heurnonewblock),  $-1$ otherwise. Break ties by node level\;
	   
	    \If {$cand \neq -1$}{
	         add $cand$ to $SB[i]$, remove it and its out edges from $G$\;
	    }
	    \If {$|SB[i]| \geq P$ or $cand = -1$}{
	        add new spatial block to $SB$; $i \gets i+1$\;
	    }
	}
	\textbf{return} SB\;
	\caption{Compute Spatial Blocks}\label{alg:streaming_blocks_generic}
\end{algorithm}
The heuristic comes in two variants. In the first one (\heurnewblock),  we add a node to a spatial block if its produced data volume is less than the data volume produced by the block's source(s) from which it depends (if any). 
We continue adding to the same spatial block until such node does exist or the block is full. Otherwise, we create a new spatial block, and we start filling it.
In this way, we are guaranteeing that the streaming interval of the block's sources is not increased by adding an upsampler node producing more data than the source itself, and no other node is slowed down by this. Note that in this case, a spatial block may have less than $P$ tasks. 
The second variant (\heurnonewblock), relaxes the requirement on the produced data volume: if no other candidate is available, a node can be added to the current spatial block even if it is producing more data than the block's source(s). In this case, all spatial blocks (except the last one) contain $P$ tasks.

In both variants, we guarantee by construction that there are no cyclic dependencies between spatial blocks: at any time, we consider candidate nodes whose predecessors have been already inserted into a spatial block. 
The proposed heuristic loops over all nodes in the graph, selecting on each step the most convenient one, according to the considered variant.
Its complexity is $\mathcal{O}(N^2)$, where $N$ is the number of nodes in the graph. 

Once the task graph is partitioned into spatial blocks, we can schedule them one after the other in the same order in which they are created, using the approach described in \Cref{sect:schedule_in_sb}.

\section{Buffer space for deadlock-free execution}\label{sect:buffer_space}
Streaming communications can exploit the Network-on-Chip, or rely on communication channels implemented in backing memory. In both cases, we abstract them as FIFO channels having a fixed buffer space, and using blocking-after-service semantics (writes may block if the FIFO buffer is full). Insufficient FIFO buffer space can cause a  deadlock even if the task graph is acyclic. This section discusses how to detect such situations and dimension FIFO channels accordingly.

Let us consider the two examples shown in~\Cref{fig:deadlocks}.
\begin{figure}[b]
    \centering
    \begin{subfigure}{0.9\columnwidth}
        \includegraphics[width=\linewidth]{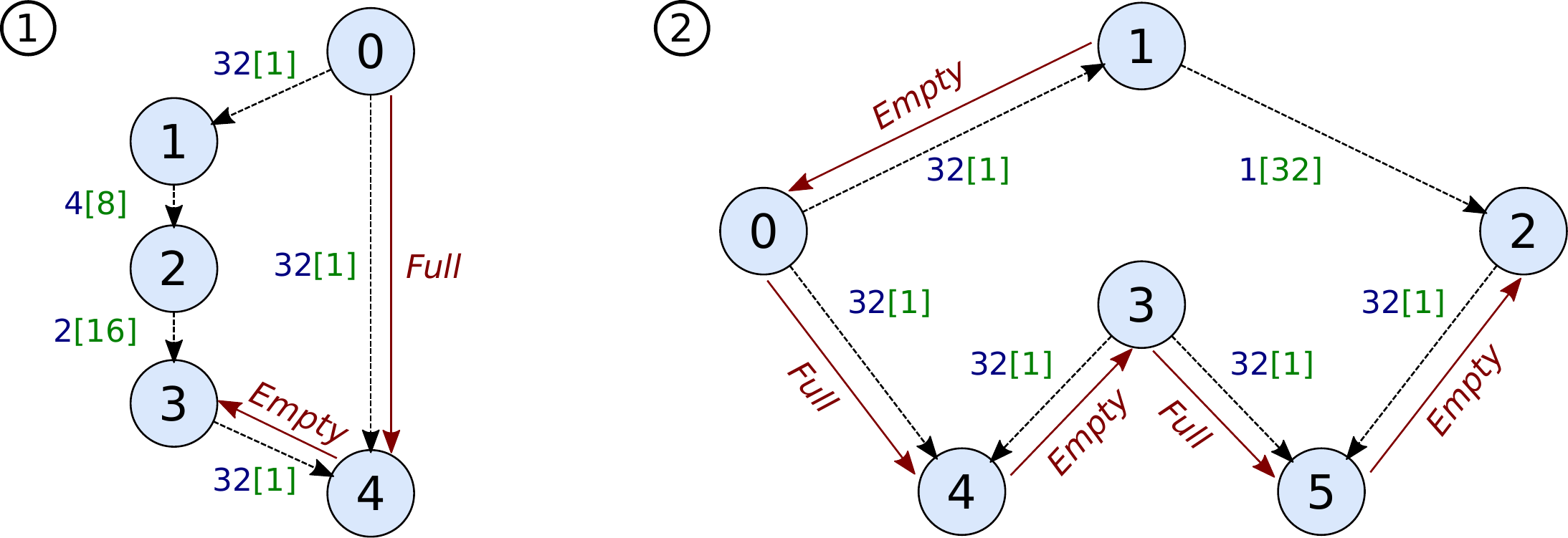}
    \end{subfigure}

    \bigskip
    \begin{subfigure}{0.45\columnwidth}
        \centering
        \footnotesize
        \begin{tabular}[t]{cccc}
        \toprule
        \textbf{Task} & $\mathbf{ST}$ & $\mathbf{LO}$ & $\mathbf{FO}$ \\
        \midrule
        0    & 0     & 32  & 1  \\
        1    & 1     & 33  & 9  \\
        2    & 9     & 34  & 18 \\
        3    & 18    & 50  & 19 \\
        4    & 19    & 51  & 20 \\
        \bottomrule
        \end{tabular}
         \end{subfigure}
    \quad
    \quad
    \begin{subfigure}{0.45\columnwidth}
        \footnotesize
        \begin{tabular}[t]{cccc}
        \toprule
        \textbf{Task} & $\mathbf{ST}$ & $\mathbf{LO}$ & $\mathbf{FO}$ \\
        \midrule
        0    & 0     & 32  & 1  \\
        1    & 1     & 33  & 33  \\
        2    & 33    & 65  & 34   \\
        3    & 0     & 32  & 1   \\
        4    & 1     & 33  & 2   \\
        5    & 34     & 66  & 35   \\
        \bottomrule
        \end{tabular}

    \end{subfigure}
\caption{Two task graphs that can deadlock because of insufficient buffer space and their schedule. Blue edge labels represent the data volume sent between tasks. Green labels indicate the streaming intervals. Red arrows highlight deadlock situations due to channels being empty or full. }
    \label{fig:deadlocks}
\end{figure}
Task graph \one illustrates a situation where a deadlock can occur because multiple disjoint paths exist between two given nodes (0 and 4 in the example). 
When task 0 sends its first element to task 4 (right path), this is waiting for the first element comings from 3 (left path). This will arrive later due to multiple reducer nodes in the left path. If the communication channel between task 0 and task 4 has insufficient buffer space, task 0 will eventually stall because the channel to task 4 gets full and cannot send more data. This will prevent the data from continuing to travel on the left path because the reducer does not receive enough data to produce its first output element.
Task graph \two in \Cref{fig:deadlocks} shows a more general situation where the nodes are part of an \textit{undirected} cycle. In this case, task 2 is waiting for the input data coming from task 1. This will let task 5 stall because it is not receiving any data. This situation will be propagated to tasks 3, 4, and 0, causing the entire computation to deadlock.
In both examples, if enough buffer space is provided for the inter-task communications, we can tolerate the delay in the generation, consumption, and propagation of data across various paths, resolving the deadlock situations. 
\begin{figure*}
    \centering
    \includegraphics[width=\textwidth]{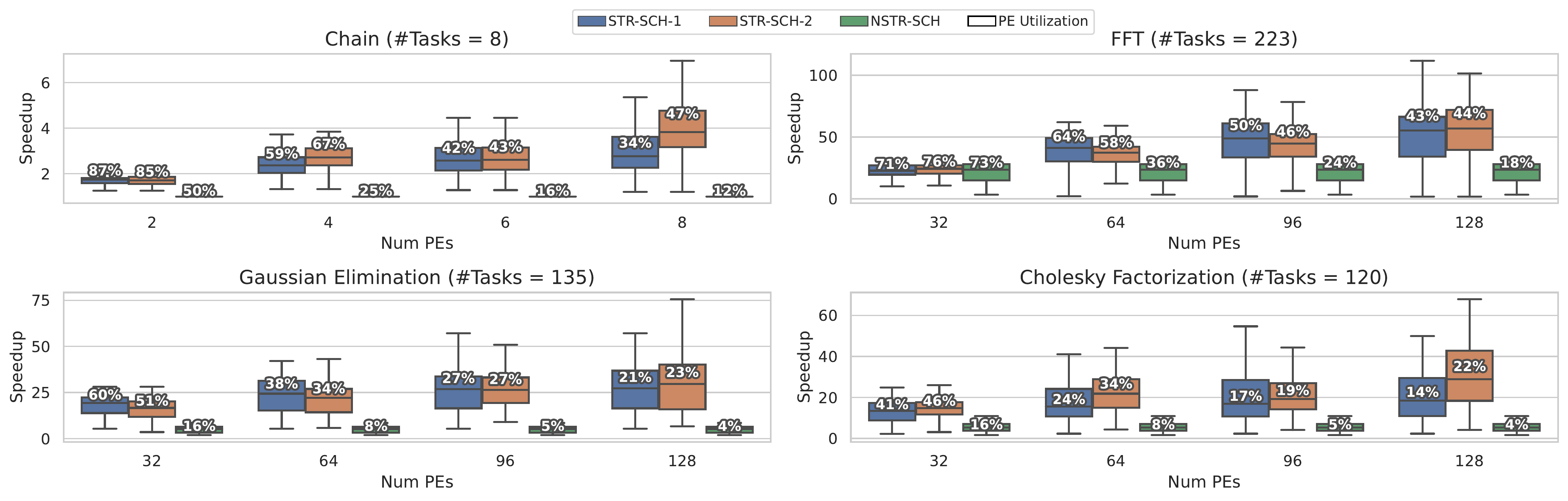}
    \vspace{-2em}
    \caption{Distributions of speedup over sequential execution for synthetic task graphs, considering streaming (\schednewblock\  and \schednonewblock\ for \heurnewblock\ and \heurnonewblock\ variants) and non-streaming scheduling (\nssched). White labels report PEs utilization.}
    \label{fig:speedup_synthetic}
\end{figure*}

In general, once the schedule has been computed, we are interested in understanding what is the \textit{smallest} buffer space required to avoid deadlocks \textit{and} guarantee that the execution behaves as expected in the computed schedule (i.e., no bubbles in pipelined communications as data flows according to the streaming intervals).
Since deadlocks can occur only along streaming paths, we can analyze each spatial block independently by considering its undirected cycles.

Let $v$ be a node in an undirected cycle having more than one predecessor in the same spatial block $B_i$.
To compute the buffer space of one of its incident streaming edges, we need to consider the highest delay a single data element experiences in reaching node $v$ through that edge. Then the buffer space is given by that delay, subtracted from the highest delay found across \textit{all} input edges of node $v$. 
For canonical task graph, the highest delay for some data reaching node $v$ through edge $(u,v)$ is given by $\FO{u}$. 

Therefore, we can compute the buffer space for its incident edges $(u,v)$ as:
\begin{equation}\label{eq:buffer_space}
B(u,v) = \frac{\max_{(t,v) \in G[B_i]} \FO{t} - \FO{u}}{\SInode{u}}
\end{equation}
or the edge data volume, if the computed buffer space is larger than the data being sent between $u$ and $v$. The denominator takes into account streaming intervals greater than 1 (i.e., the buffer space is filled at a slower rate).

It follows that for the task graph \one (\Cref{fig:deadlocks}) the FIFO channel used for the streaming communication between tasks 0 and 4 must have a buffer space equal to 18. In task graph \two, the buffer space for the channel between tasks 3 and 5 must be equal to 32.

Regarding the detection of undirected cycles, 
it is worth noting that we do not need to look at \textit{all} undirected cycles. Instead, we need to analyze all the nodes that are part of \textit{an} undirected cycle.
To find these nodes, we use a modified Depth-First Search (DFS) visit: starting from a node, we perform a DFS visit of the Spatial Block, where we do not consider the directionality of edges. A back-edge between node $u$ and $v$ indicates the presence of a cycle. We mark all the ancestors of both $u$ and $v$, until a common ancestor is found.
Once the DFS is completed, we return the weakly connected components of the marked nodes as different undirected cycles. The complexity of this approach is $O(|V|+|E|)$.

\section{Evaluation}

We implemented our analysis passes, heuristics, and buffer space computation in a proof-of-concept framework written in Python \footnote{The framework is available at: \url{https://github.com/spcl/streamingsched}}. 
Given a canonical task graph and the number of available PEs, it produces a streaming scheduling for the considered architecture and the required FIFO buffer space to prevent deadlocks.
All the tests are executed on a machine running Ubuntu 20.20, with 128 GB of main memory, and a 16C/32T AMD Ryzen 9 5950X CPU.

We experiment with two sets of graphs. First, we consider small and medium-sized random synthetic task graphs generated from four well-known computations: Tasks Chain, Fast Fourier Transform ~\cite{fft}, Gaussian Elimination \cite{heft}, and Tiled Cholesky Factorization \cite{cholesky}. 
The Chain task graph is composed of $N$ tasks, where task $i$ receives data from task $i-1$ (if present) and sends the data to task $i+1$ (if present).
The Fast-Fourier Transform graph is obtained from the one-dimensional FFT Algorithm \cite{fft, heft}, which is composed of recursive calls and butterfly operations. Given the number of input points $N$, there are $2N-1$ recursive call tasks and $N\log_2{N}$ butterfly operation tasks.
In the Gaussian Elimination \cite{hypertool, heft}, being $M$ the matrix size, the total number of tasks is equal to $\frac{M^2+M-2}{2}$.
Finally, we consider the task graph obtained for the Cholesky decomposition by using the left-looking tiled variant \cite{cholesky}. Being the matrix composed by $T\times T$ tiles, the total number of tasks is $\frac{T^3}{6} + \frac{T^2}{2} + \frac{T}{3}$.

For a given topology, we consider different DAGs by randomly generating edge weights: therefore, each task graph will have different data volumes and types of canonical nodes. We do not introduce buffer nodes so that all edges can be streaming within a spatial block. 
Then, we compare results obtained from canonical task graphs with ones obtained from a related model. Finally, we consider larger graphs representing real-world applications (Machine Learning workloads).

\paragraph{Comparison metrics}
In the following, we show the results obtained with the two variants of our steaming schedule  heuristic defined in \Cref{sect:scheduling_generic_task_graph} (\schednewblock\ and \schednonewblock\ indicating the \heurnewblock\ and \heurnonewblock\ versions, respectively).
To evaluate the gain in performance and PE utilization, we compare them with the case where all communications are buffered. We refer to this case as \emph{non-streaming} scheduling (\nssched). For this, we consider a classical critical path list-based scheduling for homogeneous processing elements, with bottom-level priorities (similar to CP/MISF, \cite{cp_misf}), and insertion slot.
To compare the obtained results across different topologies, we measure the schedule length (\textit{makespan}) and compute the following metrics:

\begin{description}[style=unboxed,leftmargin=1em]
    \item[Speedup] The speedup is defined as the ratio of the sequential execution time (computed by assigning all tasks to a single PE) to the parallel execution time (the makespan). 
    
    \item[Streaming Scheduling Length Ratio] We extend the definition of Scheduling Length Ratio (SLR) of Topcuoglu et al.~\cite{heft}, considering streaming communications. The Streaming SLR (SSLR) is defined as the ratio between the makespan and the streaming depth of the DAG.

\end{description}
\vspace{-1em}

\begin{figure*}
    \centering
    \includegraphics[width=\textwidth]{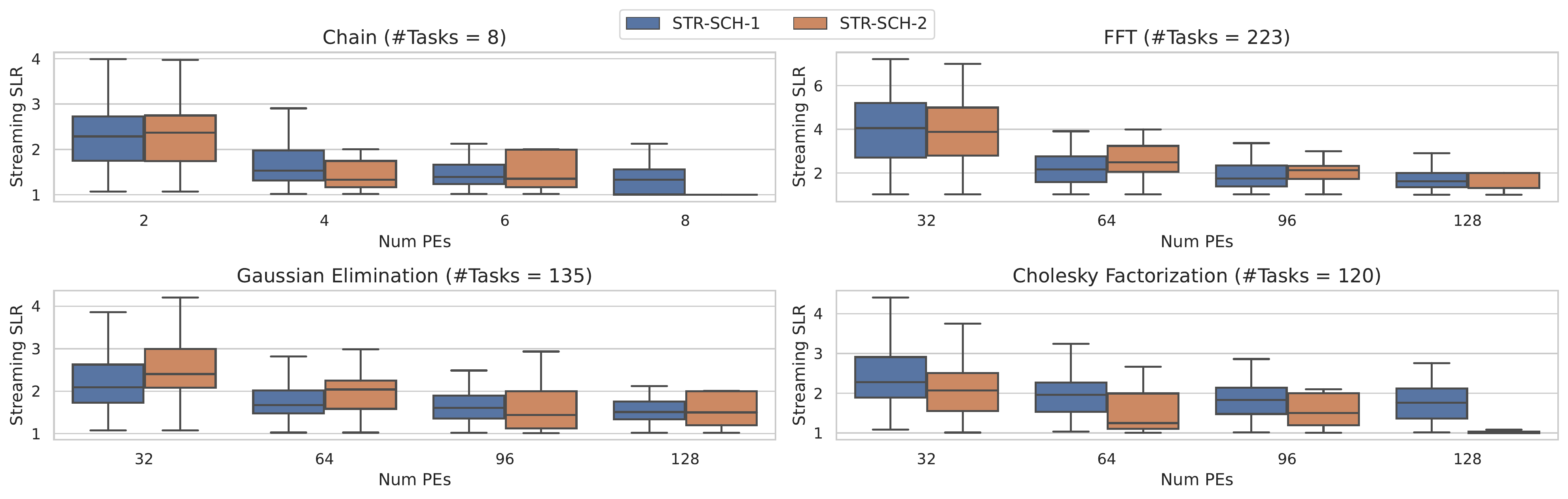}
    \vspace{-2em}
    \caption{Streaming SLR distributions for synthetic task graphs with \schednewblock\ (\heurnewblock) and \schednonewblock\ (\heurnonewblock) variants.}
    \label{fig:streaming_slr_synthetic}
\end{figure*}

\begin{figure*}
    \centering
    \includegraphics[width=\textwidth]{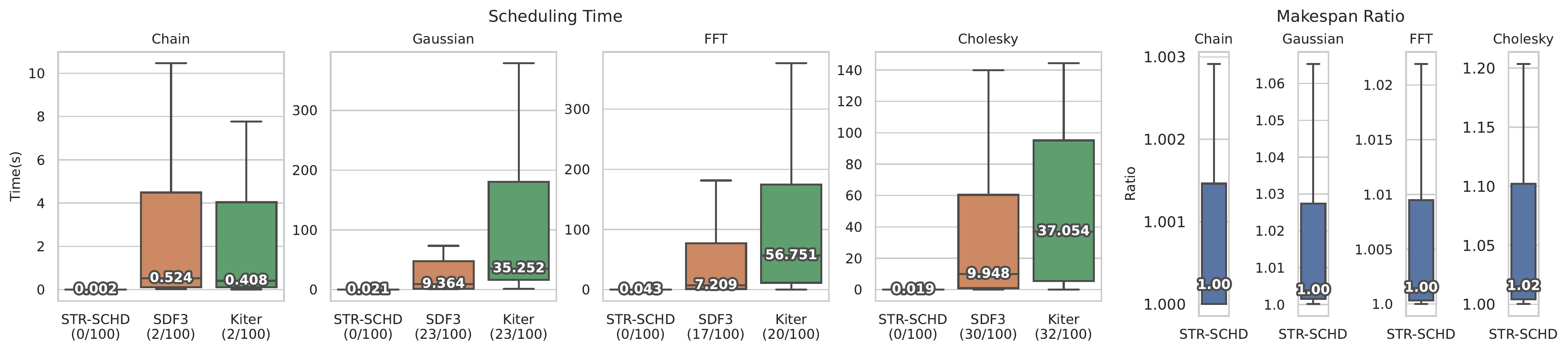}
    \vspace{-2em}
    \caption{Left: distributions of scheduling time for our approach, SDF3, and Kiter. In parenthesis, the number of graphs that timed-out over the total. Right: ratio of makespan computed with our approach over the one computed by SDF3. White labels report median values.}
    \label{fig:comparison}
    \vspace{-1em}
\end{figure*}

\subsection{Synthetic Canonical Task Graphs}\label{sect:exp_synthetic}

\Cref{fig:speedup_synthetic} shows the speedup distributions as box plots for 100 randomly generated task graphs, when scheduled considering a varying number of PEs. The middle line of a box plot represents the median; the upper and lower limits of the box indicate the first quartile (Q1) and the third quartile (Q3). White labels report the PEs utilization.

The chain task graph is composed of a linear chain of 8 tasks. Here, the non-streaming scheduling has a speedup of 1, given that all the tasks must be performed sequentially, one after the other. Instead, pipelined communications allow the concurrent execution of several tasks, giving us a higher speedup as long as we increase the number of PEs. 
Similar results are also obtained with the other considered topologies: the buffered communication approach stops scaling, while streaming scheduling enables additional gains with better PE utilization. It is worth noting that, in all the cases, the non-streaming heuristics achieves the highest attainable speedup (the corresponding SLR is 1).
Concerning the differences between the two variants of streaming scheduling, we notice that \heurnonewblock\ achieves higher speedup when the number of PEs is approaching the number of tasks in the DAG. The rationale is that this variant partitions the graph in a less or equal number of spatial blocks (one if the number of tasks is less or equal to the number of PEs) compared to the \heurnewblock\ one. Although upsampler nodes can increase the streaming intervals, scheduling a single spatial block (instead of two, or more, executed back-to-back) can sometimes pay off, producing a higher performing schedule.

\Cref{fig:streaming_slr_synthetic} reports the SSLR distributions for the two variants of the streaming scheduling heuristics.
In both cases, the SSLR is reduced as long as we increase the number of PEs.
Following the previous results, \heurnonewblock\ can approach minimum SSLR (one) for the DAGs when the number of PEs is equal to or larger than the number of tasks.

\subsection{Comparison with Synchronous DataFlow Graphs}\label{sect:exp_comparison_sdfg}

Synchronous DataFlow Graphs (SDFGs) \cite{synchronous_data_flow_graph}, and immediate extensions, such as Cyclo-static Dataflow Graphs (CSDFGs) \cite{csdfg}, have been used to analyze and schedule streaming graphs. To the best of our knowledge, they are the most related models of computation to canonical task graphs.

In an SDFG, nodes represent computations, and edges represent FIFO channels. A node can execute (``fire'') when there is enough data on all the inputs. Each edge $u \rightarrow v$ is annotated with the production rate (how many tokens $u$ produces per firing), consumption rate (how many tokens $v$ consumes per firing), and the initial number of tokens.
In a CSDFG, the number of tokens consumed and
produced by an actor varies from one firing to the next, following a periodic behavior. 
In this section, we compare our scheduling results with the ones obtained with CSDFGs.

Provided that there are no buffer nodes (not supported in CSDFGs), we can convert a given canonical task graph into an equivalent CSDFG: each canonical node is represented by a corresponding CSDFG node. Using different production/consumption rates per firing, we conveniently represent downsamplers and upsamplers.

A CSDFG can be statically analyzed to compute its optimal throughput: i.e., the number of iterations of the entire graph that can be executed per unit of time. The exact determination of the throughput requires computing an optimal schedule \cite{kiter_bodin}.
We allow only one instance of the graph to be in execution at a given time, by adding in the equivalent CSDFG edges from the sink(s) to the source(s), with an initial token. In this way, by computing the inverse of the CSDFG throughput, we derive the makespan of the implied optimal schedule, and we compare it with the one obtained by our heuristics.

We consider two publicly available CSDFG frameworks: SDF3 \cite{sdf3} and Kiter \cite{kiter}. SDF3 uses symbolic execution \cite{sdfg_throughput}, while Kiter uses K-Periodic Scheduling \cite{kiter_bodin}.
Both approaches compute the optimal throughput. As the analysis of a CSDFG is computationally expensive, we set a time-out to 1 hour for analyzing a single graph.
Since the two frameworks do not allow restricting the number of used processing elements, we set it to the number of nodes in the graph in our scheduling strategy, and we use the \heurnonewblock\ heuristic.

\Cref{fig:comparison} shows the result obtained with the same synthetic graphs generated in \Cref{sect:exp_synthetic}.
We note that the number of CSDFGs that can not be scheduled within the time constraint grows (up to 30\%) as the graph grows in complexity and size. In addition to this, in the cases where the CSDFG frameworks can return a result within the given time, the analysis time is still 2-3 orders of magnitude more expensive than with canonical task graphs.
On the left side of the figure, we report the ratio of the makespan computed with canonical graphs and the makespan computed using SDF3 (Kiter produces identical results), showing that the difference is negligible in most cases.

Therefore, when a computation can be analyzed with both approaches, canonical task graphs can produce a schedule marginally less efficient than CSDFG but in a fraction of their time, allowing the analysis of larger and more complicated applications.

\subsection{Real task graphs}

In addition to synthetic graphs, we evaluate the benefits of our approach to real-world machine-learning inference workloads.
We use DaCeML~\cite{daceml} to extract a first version of the task graph for each considered workload, where nodes are ONNX~\cite{ONNX} operators, and edges are labeled with data movements between the different operations.
From the ONNX graph, we generate the corresponding canonical task graph. Given an ONNX operation, we can distinguish between:
\begin{itemize}
    \item operators such as \texttt{Reshape}, \texttt{Transpose}, and \texttt{Slice}, that can be represented as buffer nodes.
    \item Operations that can be mapped one-to-one to a canonical task. For example, \texttt{Add}, \texttt{Sub}, and \texttt{Relu} can be mapped to element-wise tasks, while \texttt{MaxPool} and \texttt{ReduceSum} can be mapped to downsampler tasks.
    \item More complicated operators such as \texttt{MatMul}, \texttt{SoftMax}, and \texttt{Conv} must be explicitly represented as a canonical task graph as discussed in \Cref{sect:generic_dags}.
\end{itemize}
We considered the canonical task graph originated from the Resnet-50 model~\cite{resnet} and from an encoder layer of the base transformer model proposed by Vaswani et al.~\cite{transfomer}. 
We considered the \textit{im2col} approach~\cite{im2col} to express the convolution operation as a matrix-matrix multiplication. The \texttt{Conv} and the \texttt{Batch Normalization} operator in Resnet-50, and the \texttt{Softmax} and \texttt{MatMul} operators of transformer encoders, are converted into canonical nodes as showed in \Cref{sect:generic_dags}. For each \texttt{MatMul} we choose the implementation that maximizes parallelism depending on the input matrices' sizes.

The resulting canonical task graph for Resnet-50 is composed by 54,252 nodes, 246 of which are buffer nodes, while the task graph for the transformer encoder, is composed by 4,748 nodes, 37 of which are buffer nodes. Table~\ref{tab:ml_results} reports the achieved speedups and gains over the non-streaming scheduling. We do not noticed relevant difference between the two variants of the streaming scheduling, therefore we report the result for the \heurnewblock\ version.
\begin{table}[t]
\begin{minipage}[t]{.49\columnwidth}
\resizebox{\textwidth}{!}{
\begin{tabular}[t]{cccc}
\toprule
\textbf{\#PEs} & \textbf{\thead{\ssched \\ Speedup}} & \textbf{\thead{\nssched \\Speedup}} & \textbf{G} \\
\midrule
512	&	109.4	&	83.6	&	1.3\\
1024	&	123.2	&	88.3	&	1.4\\
1536	&	128.8	&	90.1	&	1.4\\
2048	&	135.0	&	90.2	&	1.5\\
\bottomrule
\end{tabular}
}
\end{minipage}
\begin{minipage}[t]{.49\columnwidth}
\resizebox{\textwidth}{!}{
\begin{tabular}[t]{cccc}
\toprule
\textbf{\#PEs} & \textbf{\thead{\ssched \\ Speedup}} & \textbf{\thead{\nssched \\Speedup}} & \textbf{G} \\
\midrule
256	&	153.7	&	111.7	&	1.4\\
512	&	218.8	&	142.5	&	1.5\\
768	&	290.6	&	149.4	&	1.9\\
1024	&	305.0	&	153.0	&	2.0\\
\bottomrule
\end{tabular}
}
\end{minipage}
\caption{Results for Resnet-50 (left) and a transformer encoder layer (right). Last column shows the performance gain (G) of the streaming scheduling over non-streaming one.}
\label{tab:ml_results}
\vspace{-1em}
\end{table}
For Resnet, we can take advantage of pipelined communications mainly between \texttt{Batch Normalization}, \texttt{ReLu} and \texttt{MaxPool} operations, resulting in a performance gain over the non-streaming version.
For the transformer encoder, the gain is higher, due to the presence of longer chains of operators that can be easily pipelined.
For both the considered applications we approach a Streaming SLR of 1 as we increase the number of PEs.

\section{Related work}
Static task graph scheduling for homogeneous processing elements is a well-studied problem in computer science. Being an inherently hard problem, various heuristics have been proposed over time~\cite{survey_scheduling_1, survey_scheduling_list_cluster}. 
These can be broadly categorized into list-based~\cite{scheduling_list,  hypertool, scheduling_list_2, heft} and cluster-based techniques~\cite{dcp_scheduling, scheduling_cluster}.
%
%
Generally, these approaches assume that computation and communications costs are given as input parameters, and that a task can only start when all its parents have terminated. 
In contrast, we assume the computation costs are proportional to data being produced and ingested, and allow concurrent execution of communicating tasks and contribute to reducing the application makespan through pipelined communication.


Traditional scheduling heuristics usually have a \textit{local view}, and decisions are made based on the graph's portion being analyzed. 
Other approaches try to make \textit{global} decisions. This is the case of look-ahead heuristics, such as the Dynamic Critical Path algorithm proposed by Kwok and Ahmad~\cite{dcp_scheduling}, or by approaches using partitioning-assisted list-based heuristics as in the work of Özkaya et al.~\cite{ipdps19_partitioning}.
Similarly, we consider the graph's global structure by partitioning it into spatial blocks and then scheduling each block independently. 
The work of Cong et al.~\cite{cong} deals with mapping streaming applications on FPGA, optimizing communication and computation simultaneously. The authors considered spatial scheduling while we simultaneously deal with both temporal \emph{and} spatial scheduling.


Synchronous DataFlow Graph (SDFG) ~\cite{synchronous_data_flow_graph}, and extensions, are the most closely related models to canonical task graphs. 
Due to their analyzability, SDFGs are commonly used for multimedia and real-time applications, and streaming languages, such as LUSTRE~\cite{lustre} and StreamIt~\cite{streamit}, are based on this model of computation. 
%
Various works tackle the problem of deadlock-freedomness and buffer sizing for SDFG graphs. 
The buffer sizing problem is NP-complete~\cite{sdfg_np_complete_buffer}, 
and approximate or heuristic solutions have been proposed. 
Stuijk et al.~\cite{sdfg_1, sdfg_2} present a heuristic approach for computing the complete trade-off space between the throughput and buffer size of a given SDFG. Li et al.~\cite{sdfg_3} analyze different types of deadlocks and propose solutions to deal with them.

These approaches can not be generally applied to canonical task graphs, as they follow a different model of computation.
We highlight several fundamental differences between SDFGs and canonical task graphs:
\begin{itemize}
\item Previous work on scheduling SDFGs mainly target pipelining across multiple graph iterations, with the throughput being the primary concern. Instead, we propose a model that considers pipelining across tasks (i.e., within a single graph iteration) as a first-class citizen, and we focus on optimizing the latency of a single graph iteration.
\item  SDFGs can only be used to model fully streaming applications. We can also explicitly represent non-streaming [sub-] computations thanks to the buffer node concept.
\item  We provide easy-to-compute bounds on the application's parallel and streaming execution time. To the best of our knowledge, similar bounds exist only for simpler variants of SDFGs (e.g., Homogeneous DFG). More interesting graphs need to be transformed into HDFG (as in Cong et al. ~\cite{cong}), resulting in a graph having a size that, in the worst case, is exponentially larger than the original graph.
\end{itemize}

Various commercially available accelerators allow the user to take advantage of streaming communication either through low-level APIs or via proprietary compiler passes: for example the Sambanova Reconfigurable Data Flow Architecture~\cite{sambanova}, Xilinx ACAP devices~\cite{versal}, and Cerebras Wafer Scale chips~\cite{cerebras}.
Despite these specific solutions, we believe that a complete methodology for dealing with the problem of scheduling a streaming computation on dataflow architecture is yet to be established. In this work, we contribute with a holistic view of the problem, abstracting away from the underlying hardware characteristics, and proposing solutions for the spatio-temporal scheduling of applications on homogeneous processing elements.

\section{Conclusion}

This paper proposes methods, analyses, and algorithms to schedule task graphs on dataflow architectures, explicitly considering task pipelining and their concurrent spatial execution. The analysis at steady-state facilitates the reasoning and enables decisions that take into account the global structure of the task graph and its dataflow characteristics. This allows us to partition the task graph into temporally multiplexed components of spatially executed tasks and to compute buffer size to guarantee deadlock freedom.
With streaming scheduling, we can better exploit a dataflow device, increasing the speedup over non-streaming approaches, even for large graphs.

The proposed method can be extended considering dataflow architectures with heterogeneous processing elements,  typical of System-on-Chip, and taking into account placement, which plays a crucial role in Coarse-Grained Reconfigurable Arrays.
While pipelining across multiple running tasks is a natural fit for dataflow architectures, we believe this approach can be applied to other platforms or to clusters of dataflow devices. The proposed models and analysis are still relevant but need to be adequately extended, for example by considering communications across devices.

\section*{Acknowledgments}
This work was supported by the ETH Future Computing Laboratory (EFCL), financed by a donation from Huawei Technologies, and by the European Research Council (ERC) \includegraphics[height=1em]{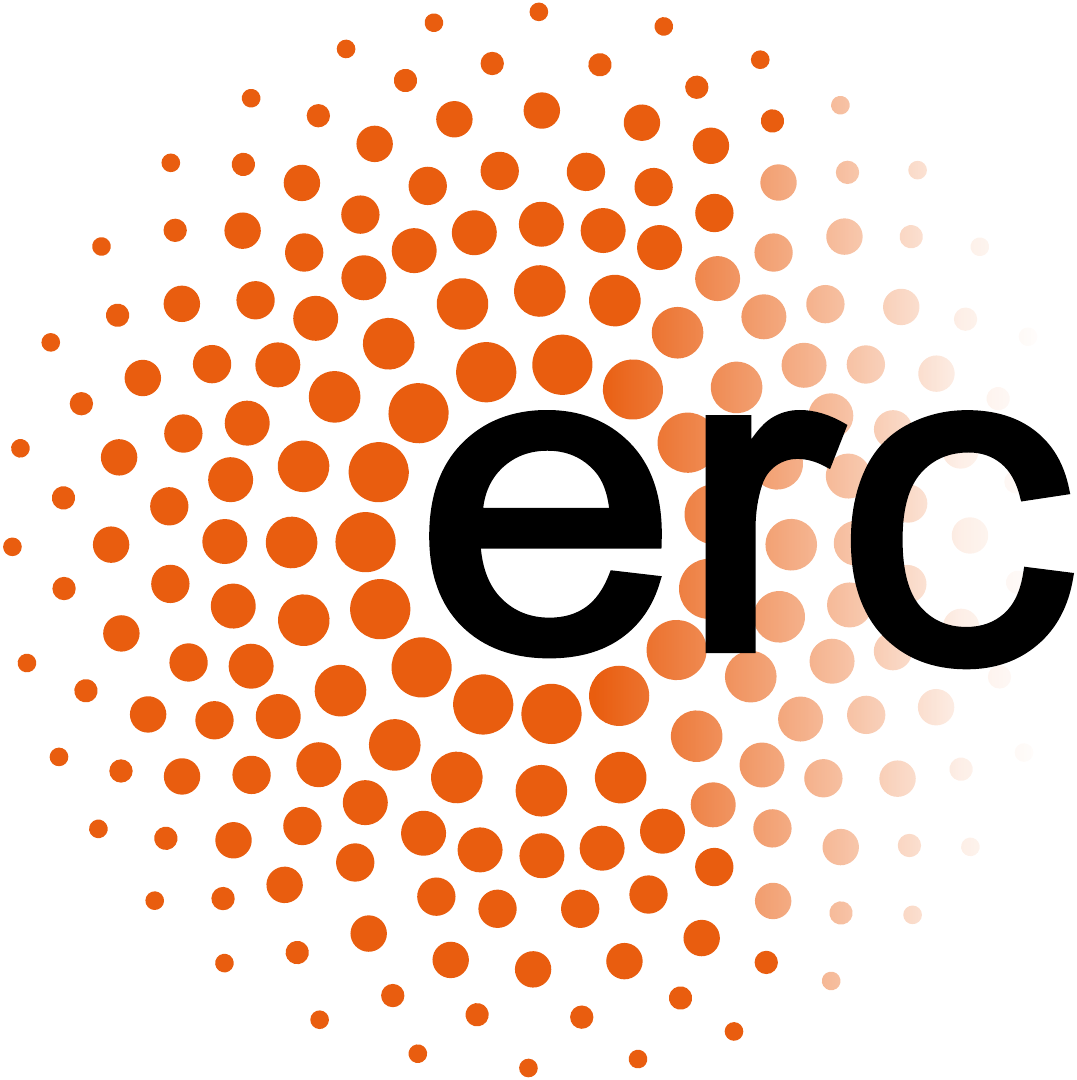} Project PSAP (grant No. 101002047).

\bibliography{refs}

\appendix

\section{Scheduling - Spatial Block Partitioning}\label{suppl:spatial_block_partitioning}
In this section we discuss the spatial block partitioning of canonical task graphs with given characteristics, and we show how we can derive bounds for their execution.

\subsection{Element-wise Task Graph}
Consider the case of a task graph composed only of element-wise tasks.
In this case, we can construct a schedule for $P$ processing elements that obtains near-linear speedup over the work until reaching the streaming depth in the limit. The statement has the same form as Brent's theorem for non-streaming tasks graph~\cite{brent_lemma}.

We can order the tasks by their level, breaking ties arbitrarily. Then, using this order, we subdivide the tasks into spatial blocks of $P$ tasks. Each spatial block receives a level order according to the position in the task level order. We schedule the spatial blocks in this level order.

\begin{theorem}\label{thm:elem-wise-schedule}
The time $T_P$ it takes to execute an element-wise streaming task graph on $P$ processors satisfies $\SD \leq T_P \leq \frac{T_1}{P}+\SD$.
\end{theorem}
\begin{proof}
Consider a longest path $v_1, \dotsc v_l$ in the task DAG G. From the point on when $v_1$ gets scheduled, it takes $k$ steps until the last element gets output by $v_1$. It takes at least one time step for this element to move along the path for the other vertices. As the path has $\level{G}$ vertices, the lower bound follows (see Section 4.2.2 of the main paper).

Consider the time it takes to complete a spatial block $B_i$. We treat the subgraph of $G$ induced by the tasks of $B_i$ as a task sub-DAG $G[B_i]$. By construction of the scheduling, once $B$ gets scheduled, all predecessors have completed and their results are available in memory. Hence, it takes $k$ time to fill up the pipeline of $G[B_i]$ and another $\level{G[B_i]}-1$ steps for the last element to finish the pipeline.

Now, we sum the total time:
\begin{small}
\begin{align*}
T_p &= \sum_{i=1}^{\lceil n / p \rceil} (k + \level{G[B_i]} - 1) = k \lceil n/p \rceil + \sum_{i=1}^{\lceil n / p \rceil} (\level{G[B_i]} - 1) \\
&\leq  k \lceil n/p \rceil + \level{G} - 1 \leq  \frac{k n}{p} + k + \level{G} - 1 
= \frac{T_1}{p}+T_{\infty}^{s}
\end{align*}
\end{small}
The only step in the derivation that is not arithmetic is the relation between the sum of the levels of the spatial blocks and the overall number of levels of the graph: $ \sum_{i=1}^{\lceil n / p \rceil} (\level{G[B_i]} - 1 \leq \level{G} - 1$. To see the bound, we charge the costs of the spatial blocks to the levels of the task DAG G: when a spatial block does not contain tasks with different levels (in G), then it does not contribute anything to the cost. Otherwise, consider the tasks of block $B_i$ in level order. Whenever we change from level $i$ to level $i+1$, we charge this cost to level $i$ of $G$. No level of $G$ is double-charged this way because we constructed the spatial blocks level-wise, and the last level is never charged.
\end{proof}

\begin{figure*}
    \centering
    \includegraphics[width=\textwidth]{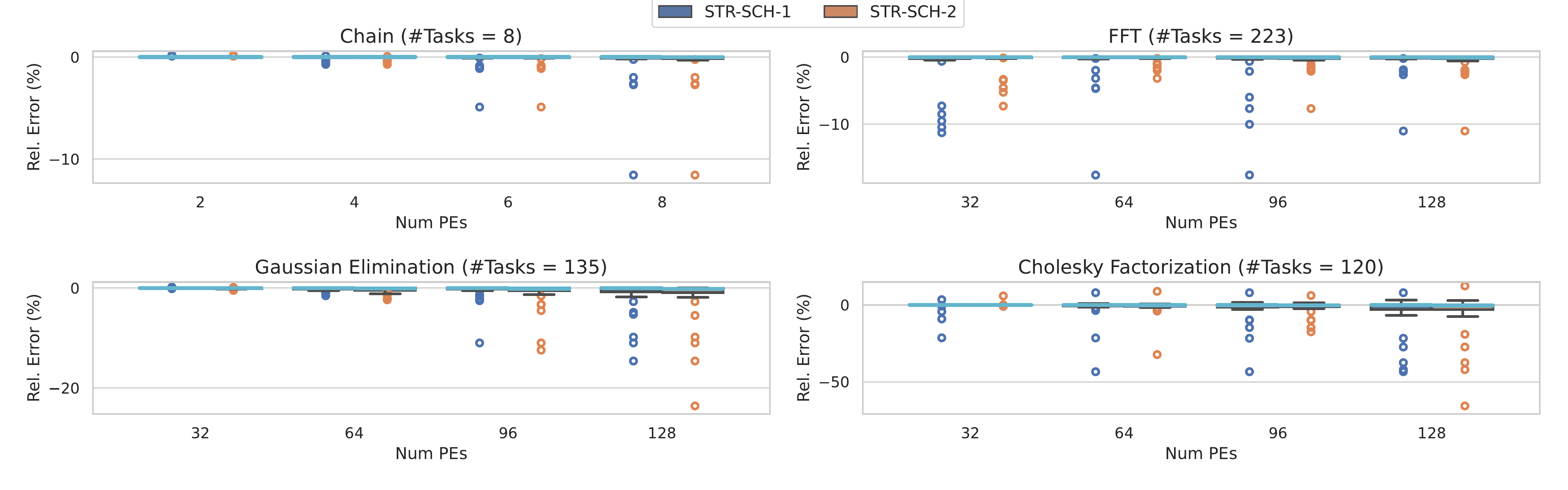}
    \caption{Error results for synthetic task graphs with \schednewblock\ (\heurnewblock) and \schednonewblock\ (\heurnonewblock) variants.}
    \label{fig:error_synthetic}
\end{figure*}

\subsection{Downsampler Task Graph}\label{suppl:downsampler_partitioning}

Consider the case of a task graph composed only of element-wise and downsampler nodes.
We define the spatial blocks by grouping together nodes that have similar work, proceeding in non-increasing order of work as shown in \figurename~\ref{fig:downsampler_dag}. We detail the algorithm in Algorithm~\ref{alg:streaming_blocks_downsampler}.
\begin{figure}[h]
    \centering
    \includegraphics[width=6cm]{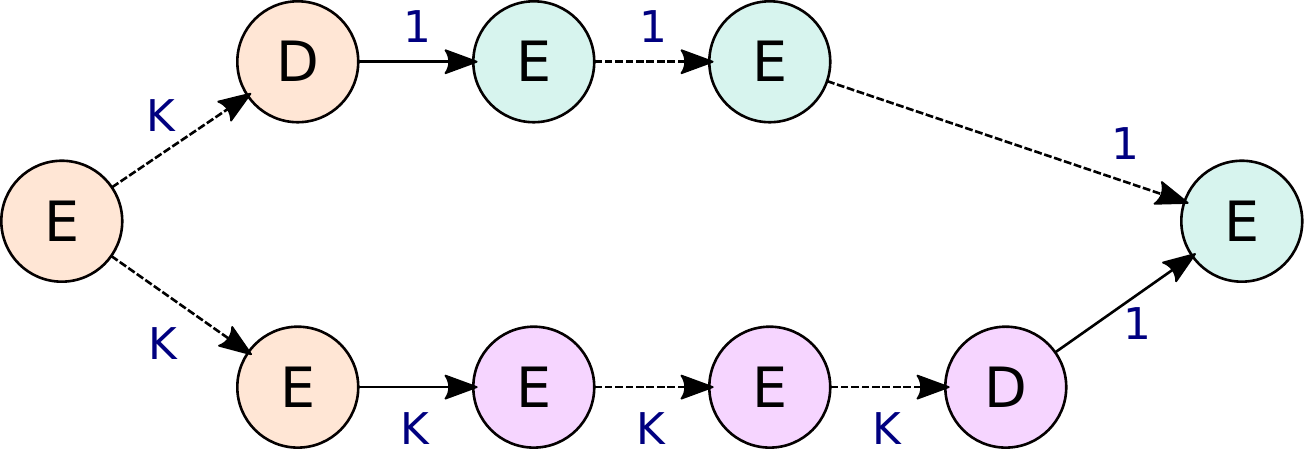}
    \caption{Task graph composed by element wise and downsampler nodes. Blue labels indicate data volumes. Different node colors identifies different spatial blocks computed considering $P=3$.}
    \label{fig:downsampler_dag}
\end{figure}
\begin{algorithm}[h]
	\footnotesize
	\SetAlgoLined
	\KwIn{Task Graph $G = (V,E)$ composed by elwise and downsampler nodes, number of PEs $P$}
	\KwResult{a partition of $G$'s nodes in spatial blocks}
	Compute the work for each node in the graph\;
	$SB \gets [\{\emptyset\}]$\;
	$i \gets 0$\;
	
	\While{$|V|>0$}{
	    $sources \gets$ compute source nodes of the graph\;
	    $cand \gets$ select node in $sources$ with highest work (and lowest level in case of tie)\;
	    \If {$|SB[i]| \geq P$}{
	        add new spatial block to $SB$\;
	        $i \gets i+1$\;
	    }
	    $SB[i] = SB[i] \cup \{cand\}$\;
	    remove $cand$ and its outgoing edges from $G$ \;
	}
	
	\textbf{return} SB\;
	\caption{Compute Spatial Blocks}\label{alg:streaming_blocks_downsampler}
\end{algorithm}

\begin{theorem}
Consider a task DAG $G$ with element-wise and downsampler nodes. Let $x$ be the maximum number of nodes in the same level with different work. The proposed algorithm satisfies $$T_p \leq \frac{T_1}{p} + \SD + \min(n-1,  (x - 1)(\level{G}-1)) \enspace .$$
\end{theorem}
\begin{proof}
Observe that the algorithm picks nodes in order of non-increasing work. This is because, along any path in the task DAG, the work only stays the same or decreases. Hence, once we schedule a node, it cannot happen that a node with larger work than was last scheduled becomes available to schedule.

Hence, we can consider the nodes ordered by non-increasing work $v_1, \dotsc, v_n$, partitioned into blocks of $p$ nodes (except for the last block). Let $B_i$ denote the $i$-th spatial block and let $\max(B_i)$ be the maximum work of any node in the $i$-th spatial block. Note that the time it takes to execute $B_i$ is $\max(B_i)$ to fill the pipeline and $\level{G[B_i]}-1$ to finish streaming. Overall, the work $T_p$ is 
$$T_p =\sum_{i=1}^{\ceil{n/p}} \ \left( \max(B_{i}) + \level{G[B_i]} - 1 \right) \enspace .$$

Next, we upper bound $T_p$. Because we cannot cross levels more than $n-1$ times in total, we get that $\sum_i \left(\level{B_i}-1 \right) \leq n-1$. 
Each time a block has some node in level $i$ and level $i+1$, we charge this to level $i$. Observe that the scheduling algorithm can be viewed as scheduling the subgraphs induced by a certain work-amount in order of their weight, and then each subgraph in level-order. Hence, each level is charged at most $x$ times, except the last (which is never charged). We conclude $\sum_i \left(\level{B_i}-1\right) \leq x(\level{G}-1)$.

We compute a lower bound for $T_1$. All streaming nodes in the $i$-th block have at least as much work as $\max(B_{i+1})$. This means we can bound the work of a block by the 
maximum of the next block (times $p$). Hence, we can can define a lower bound for the work $T_1$ as follows:
$$T_1 \geq \sum_{i=2}^{\ceil{n/p}} p \ \max(B_{i})$$
Note that the sum leaves out the block $B_1$, which contains the maximum weight node.
We observe that $$\frac{T_1}{p} \geq \sum_{i=2}^{\ceil{n/p}} \max(B_{i})$$

Hence, we conclude:
\begin{align*}
T_p &\leq \min(n-1,  x(\level{G}-1)) + \sum_{i=1}^{\ceil{n/p}} \ \max(B_{i}) \\ 
&\leq \min(n-1,  x(\level{G}-1)) +  \left(\max_{v \in V(G)} \W{v} \right) + \frac{T_1}{p} \\
&\leq \min(n-1,  (x-1)(\level{G}-1)) +  T_{\infty}^s + \frac{T_1}{p} \qedhere
\end{align*}
\end{proof}

With symmetric arguments, we can draw similar conclusions for case of a canonical task graph composed only by upsampler and element-wise nodes.


\section{Evaluation}

\subsection{Validation on Synthetic Task Graphs}\label{sec:exp_simulation}

We use Discrete Event Simulation to assess the correctness of buffer space computation for pipelined communications (i.e., the simulation does not deadlock), and the quality of results (the steady-state analysis allows us to compute a realistic makespan).
The Discrete Event Simulation is implemented in \texttt{simpy}, a process-based discrete-event simulation framework based on Python.
For the simulation we take into account:
\begin{itemize}
    \item data communication volumes and dependencies, as expressed in the given task graph;
    \item communication type (streaming/non-streaming), as decided by our spatial blocks;
    \item PE assignments of each task, as decided by the scheduling heuristic.
\end{itemize}

An independent \textit{process} simulates each task of the DAG.
Streaming communications between tasks are modeled using FIFO channels. They have finite size and are accessed using blocking-after-service semantic (the sender hangs if there is no free space). FIFOs are dimensioned according to the computed buffer size.

We run the simulation, pick the simulated application makespan, and compute the relative error among the simulated application makespan and the makespan reported by our scheduling algorithm: a negative error, indicates that the scheduling makespan is larger than the simulated one.
For all the considered cases, simulations finish without deadlocks (the computed buffer space is sufficient).
\Cref{fig:error_synthetic} reports the error distribution as boxplots for the two versions of the proposed streaming scheduling heuristic. In each boxplot, the middle turquoise line represents the median, The upper and lower limits of the box indicate the first quartile (Q1) and the third quartile (Q3). The whiskers extents show the rest of the distribution, except for points that are determined to be outliers. Being $IQR = Q3-Q1$ the interquartile range, the lower and upper whiskers indicate the smallest sample $> Q1-1.5 \cdot IQR$, and the largest sample $< Q3 + 1.5 \cdot IQR$, respectively. Samples outside the whiskers are outliers and reported as circles. For readability, we reported only the top-5 and bottom-5 outliers, if present.

As it can be noted, for the considered cases, the median error is zero, or very close to zero, showing how our steady-state analysis correctly models the actual execution on average. We do not notice any sensible differences between the two considered streaming heuristics. Quartiles and whiskers are very narrow, with the greatest whiskers extent being $[-7\%, 4\%]$ for Cholesky factorization with 128 PEs. Outliers are usually negative, meaning that our analysis could underestimate the actual execution time.
The same test case exhibits the largest outliers, with a value over 50\% error, as a consequence of a more densely connected task graph.

\end{document}